\documentclass[twoside,onecolumn,12pt]{TCOM}

\usepackage{graphicx}
\usepackage{doublespace}
\usepackage{dcolumn}
\usepackage{array}
\usepackage{amsmath}
\usepackage{fancyheadings}
\usepackage{subfigure}
\usepackage{algorithm}
\usepackage{algorithmic}
\usepackage{verbatim}
\usepackage{tabularx}
\usepackage{amssymb}

\newcolumntype{d}[1]{D{.}{.}{#1}}
\newcolumntype{i}{D{.}{}{0}}

\begin{document}
\begin{spacing}{1}
\pagestyle{plain}
\date{}
\title{Minimizing Weighted Sum Download Time for One-to-Many File Transfer in Peer-to-Peer Networks}
\author{Bike Xie, Mihaela van der Schaar and Richard D. Wesel\\
Department of Electrical Engineering, University of California, Los
Angeles, CA 90095-1594\\Email:
xbk@ee.ucla.edu, mihaela@ee.ucla.edu, wesel@ee.ucla.edu
\thanks{}}

\maketitle \thispagestyle{plain}

\begin{abstract}
This paper considers the problem of transferring a file from one
source node to multiple receivers in a peer-to-peer (P2P) network.
The objective is to minimize the weighted sum download time (WSDT)
for the one-to-many file transfer. Previous work has shown that,
given an order at which the receivers finish downloading, the
minimum WSD can be solved in polynomial time by convex optimization,
and can be achieved by linear network coding, assuming that node
uplinks are the only bottleneck in the network. This paper, however,
considers heterogeneous peers with both uplink and downlink
bandwidth constraints specified. The static scenario is a
file-transfer scheme in which the network resource allocation
remains static until all receivers finish downloading. This paper
first shows that the static scenario may be optimized in polynomial
time by convex optimization, and the associated optimal static WSD
can be achieved by linear network coding. This paper then presented
a lower bound to the minimum WSDT that is easily computed and turns
out to be tight across a wide range of parameterizations of the
problem. This paper also proposes a static routing-based scheme and
a static rateless-coding-based scheme which have almost-optimal
empirical performances. The dynamic scenario is a file-transfer
scheme which can re-allocate the network resource during the file
transfer. This paper proposes a dynamic rateless-coding-based
scheme, which provides significantly smaller WSDT than the optimal
static scenario does.
\end{abstract}

\end{spacing}
\begin{spacing}{1.67}

\section{Introduction}
\label{sec:introduction}

P2P applications (e.g, \cite{BT}, \cite{Napster}, \cite{Gnutella},
\cite{KaZaA}) are increasingly popular and represent the majority of
the traffic currently transmitted over the Internet. A unique
feature of P2P networks is their flexible and distributed nature,
where each peer can act as both a server and a client \cite{AND04}.
Hence, P2P networks provide a cost-effective and easily deployable
framework for disseminating large files without relying on a
centralized infrastructure \cite{Liu07}. These features of P2P
networks have made them popular for a variety of broadcasting and
file-distribution applications \cite{Liu07} \cite{Zhang05}
\cite{Pai05} \cite{Mutualcast04} \cite{Li04} \cite{Xiang04}
\cite{Overcast}.

Specifically, chunk-based and data-driven P2P broadcasting systems
such as CoolStreaming \cite{Zhang05} and Chainsaw \cite{Pai05} have
been developed, which adopt pull-based techniques \cite{Zhang05},
\cite{Pai05}. In these P2P systems, the peers possess several chunks
and these chunks are shared by peers that are interested in the same
content. An important problem in such P2P systems is how to transmit
the chunks to the various peers and create reliable and efficient
connections between peers. For this, various approaches have been
proposed including tree-based and data-driven approaches (e.g.
\cite{Li04} \cite{Chu04} \cite{DES01} \cite{Jiang03} \cite{Cui04}
\cite{PAD03} \cite{PAD02}).

Besides these practical approaches, some research has begun to
analyze P2P networks from a theoretic perspective to quantify the
achievable performance. The performance, scalability and robustness
of P2P networks using network coding are studied in \cite{Chou04}
\cite{ACE05}. In these investigations, each peer in a P2P network
randomly chooses several peers including the server as its parents,
and also transmits to its children a random linear combination of
all packets the peer has received. Random linear network coding
\cite{NetworkInfoFlow} \cite{LinearNetworkCode}
\cite{AlgebraicNetworkCode}, working as a perfect chunk selection
algorithm,  makes elegant theoretical analysis possible. Some other
research investigates the steady-state behavior of P2P networks with
homogenous peers by using fluid models \cite{Qiu04} \cite{Ge03}
\cite{CLE04}.

In a P2P file transfer application (e.g, BitTorrent \cite{BT},
Overcast \cite{Overcast}), the key performance metric from an
end-user's point of view is the download time, i.e., the time it
takes for an end-user to download the file. In \cite{Mutualcast04},
Li, Chou, and Zhang explore the problem of delivering the file to
all receivers in minimum amount of time (equivalently, minimizing
the maximum download time to the receivers) assuming node uplinks
are the only bottleneck in the network. They introduce a
routing-based scheme, referred to as Mutualcast, which minimizes the
maximum download time to all receivers with or without helpers.

This paper also focuses on file transfer applications in which peers
are only interested in the file at full fidelity, even if it means
that the file does not become available to all peers at the same
time. In particular, this paper considers the problem of minimizing
weighted sum download time (WSDT) for one-to-many file transfer in a
peer-to-peer (P2P) network.  Consider a source node $s$ that wants
to broadcast a file of size $B$ to a set of $N$ receivers
$i\in\{1,2,\cdots,N\}$ in a P2P network.  Our model assumes that the
source uplink bandwidth constraint $U_s$, the peer uplink bandwidth
constraints $U_i$, and the peer downlink bandwidth constraints $D_i$
are the only bottlenecks in the network.  Limited only by these
constraints, every peer can connect to every other peer through
routing in the overlay network.

In order to understand the fundamental performance limit for
one-to-many file transfer in P2P networks, it is assumed that all
nodes are cooperative, and a centralized algorithm provides the
file-transfer scenario with the full knowledge of the P2P network
including the source node's uplink capacity , and the weights,
downlink capacities, and uplink capacities of peers. The cooperative
assumption holds in many practical applications, for example, in
¡°closed¡± content distribution systems where the programs are
managed by a single authority.

\section{Main Contribution}
\label{sec:contribution}

The general problem of minimizing WSDT divides into an exhaustive
set of cases according to three attributes.  The first attribute is
whether the allocation of network resources is static or dynamic. In
the static scenario, the network resource allocation remains
unchanged from the beginning of the file transfer until all
receivers finish downloading. The dynamic scenario allows the
network resource allocation to change as often as desired during the
file transfer.

The second attribute is whether downlink bandwidth constraints are
considered to be unlimited (i.e. $D_i=\infty$) or not (i.e. $D_i \le
\infty$).  Most research in P2P considers the download bandwidth
constraints to be unlimited because the uplink capacity is often
several times smaller than the downlink capacity for typical
residential connections (e.g., DSL and Cable).  However,
consideration of downlink bandwidth constraints can be important.
The downlink capacity can still be exceeded when a peer downloads
from many other peers simultaneously, as in the routing-based scheme
proposed in \cite{Wu09}.

The third attribute is whether we consider the special case of sum
download time (i.e. $W_i=1 \text{ for all } i$) or the general case
of weighted sum download time which allows any values of the weights
$W_i$.

With these cases in mind, here is an overview of the results
presented in this paper.  For the static scenario that considers
download bandwidth constraints $D_i \le \infty$ and allows any
values of $W_i$, Section \ref{sec:cvx_soln_and_bnd} uses a
time-expanded graph and linear network coding to show that the
minimum WSDT and the corresponding allocation of network resources
may be found in polynomial time by solving a convex optimization
problem.  We also present a lower bound on minimum WSDT that is
easily computed and turns out to be tight across a wide range of
parameterizations of the problem.

While the minimum WSDT for the static scenario may be found in
polynomial time using the approach of Section
\ref{sec:cvx_soln_and_bnd}, that approach is sufficiently
computationally intensive that Sections
\ref{sec:extended_mutualcast} and \ref{sec:depth2} provide lower
complexity alternatives. In some cases, the lower complexity
approaches are exactly optimal. For the remaining cases, the lower
bound of Section \ref{sec:cvx_soln_and_bnd} shows that their
performance is indistinguishable from the lower bound and hence
closely approach optimality across a wide range of
parameterizations.

Sections  \ref{sec:extended_mutualcast} and \ref{sec:depth2} build
on the foundation of the Mutualcast algorithm \cite{Mutualcast04}.
Mutualcast is a static rate allocation algorithm designed to
minimize the maximum download time to all peers in the case where
$D_i=\infty$.  Section \ref{sec:cvx_soln_and_bnd} concludes by
showing that Mutualcast achieves that section's lower bound when
$W_i=1 \text{ for all } i$ and therefore minimizes sum download time
as well as maximum download time.

Inspired by this result, Section \ref{sec:extended_mutualcast}
proposes a generalization of this algorithm, Extended Mutualcast,
that minimizes sum download time even when the download bandwidth
constraints $D_i$ are finite and distinct from each other.  When
uplink bandwidth resources are plentiful, Extended Mutualcast also
minimizes weighted sum download time regardless of weights because
each receiver is downloading content as quickly as possible given
its download bandwidth constraint and the upload bandwidth
constraint of the source.

It is notable that Mutualcast and Extended Mutualcast achieve their
optimal results while utilizing only depth-1 and depth-2 trees.
Inspired by this fact and the technique of rateless coding, Section
\ref{sec:depth2} attacks the general problem of minimizing weighted
sum download time(WSDT) by proposing a convex optimization approach
that assumes only trees of depth one or two.  Then, Section
\ref{sec:depth2} proposes a simple water-filling approach using only
depth-1 and depth-2 trees.  While the optimality of this approach is
not proven, Section \ref{sec:staticsimulation} shows that its
performance matches that of the lower bound of
\ref{sec:cvx_soln_and_bnd} for a wide variety of parameterizations.
Thus this water-filling approach provides a simple algorithm that
empirically achieves the lower bound on WSDT for all cases of the
static scenario across a wide range of parameterizations.

Turning our attention to the dynamic scenario, Wu et al. \cite{Wu09}
demonstrate that given an order in which the receivers finish
downloading, the dynamic allocation (neglecting downlink bandwidth
constraints) that minimizes WSDT can be obtained in polynomial time
by convex optimization and can be achieved through linear network
coding.  They also propose a routing-based scheme which has
almost-optimal empirical performance and demonstrate how to
significantly reduce the sum download time at the expense of a
slight increase in the maximum download time.

Dynamic schemes can reduce the minimum sum download time to
approximately half that of the static case, at least when downlink
capacities are considered to be infinite \cite{Wu09}. Essentially,
\cite{Wu09} shows that to optimize WSDT the network resource
allocation should remain constant during any ``epoch'', a period of
time between when one receiver finishes downloading and another
finishes downloading.  Thus, one optimal solution for the dynamic
scenario is ``piecewise static''.  However, \cite{Wu09} leaves the
proper selection of the ordering as an open problem and does not
address the finite downlink capacities $D_i<\infty$ or the general
case of weighted sum download time which allows any values of the
weights $W_i$.

Section \ref{sec:dynamic} provides a practical solution for the
dynamic scenario.  Specifically, it provides an approach the
ordering problem left open by  \cite{Wu09} by reformulating the
problem as that of determining the weights that should be assigned
during each static epoch so as to produce the piecewise static
solution that minimizes the WSTD (according to the original
weights).  This approach handles both finite downlink capacities
$D_i<\infty$ and the general case of weighted sum download time
which allows any values of the weights $W_i$. A key result of this
section is that, regardless of how the overall weights $W_i$ are
set, the ``piecewise static'' solution may be obtained by finding
the appropriate weights for each epoch and solving the static
problem for that epoch.  Furthermore, during any epoch the
appropriate weights of all peers are either 1 or zero with the
exception of at most one "transitional" peer whose weight can be
anywhere between zero and 1.  Neglecting the "transitional" node,
the ordering problem becomes approximately one of choosing which
peers should be served during each epoch.  Having resolved the
ordering problem in this way, the simple water-filling approach of
Section \ref{sec:depth2} provides the rate allocations for the
source and for each peer during each of the piecewise-static epochs.
Thus this section provides a complete solution for the dynamic
scenario.  Because the selection of the ordering and the rate
allocation are both close to optimal, we conjecture that the overall
performance of this solution is close to optimal across a wide range
of parameterizations.

Section \ref{sec:conclusions} delivers the conclusions of this
paper.

\section{Convex Optimization of WSDT in the Static Case}
\label{sec:cvx_soln_and_bnd}

This section considers a static P2P network in which the source node
with uplink bandwidth $U_s$ seeks to distribute a file of size $B$
so as to minimize the weighted sum of download times given a static
allocation of resources.  The static scenario assumption also
indicates that no peer leaves or joins during the file transfer.
There are $N$ peers who want to download the file that the source
node has. Each peer has weight $W_i$, downlink capacity $D_i$ and
uplink capacity $U_i$, for $i=1,2,\cdots,N$. It is reasonable to
assume that $D_i \geq U_i$ for each $i = 1,\cdots, N$ since it holds
for typical residential connections (e.g., Fiber, DSL and Cable). In
case of $D_i < U_i$ for some $i$, we just use peer $i$'s part of the
uplink capacity which equals to its downlink capacity and leave the
rest of the uplink capacity unused.

The uplink and downlink capacities of each peer are usually
determined at the application layer instead of the physical layer,
because an Internet user can have several applications that share
the physical downlink and uplink capacities. The peer weights depend
on the applications. For broadcast applications such as
CoolStreaming \cite{Zhang05} and Overcast \cite{Overcast}in which
all peers in the P2P network are interested in the same content, all
peer weights in the content distribution system can be set to 1. In
multicast applications such as ``Tribler'' \cite{Tribler} peers
called helpers, who are not interested in any particular content,
store part of the content and share it with other peers.  Assign
weight zero to helpers, and weight 1 to receivers. In some
applications, P2P systems partition peers into several classes and
assign different weights to peers in different classes.

Denote the transmission rate from the source node to peer $j$ as
$r_{s\rightarrow j}$ and the transmission rate from peer $i$ to peer
$j$ as $r_{i \rightarrow j}$. The total download rate of peer $j$,
denoted as $d_{j}$, is the summation of $r_{s\rightarrow j}$ and
$r_{i \rightarrow j}$ for all $i \neq j$. Since the total download
rate is constrained by the downlink capacity, we have
\begin{equation}
d_j = r_{s\rightarrow j} + \sum_{i \neq j} r_{i \rightarrow j} \leq
D_j, \forall j=1,\cdots,N.
\end{equation}
As a notational convenience, we also denote $r_{j\rightarrow j}$ as
the transmission rate from the source node to peer $j$ so that
\begin{equation}
d_j = \sum_{i=1}^{N} r_{i \rightarrow j} \leq D_j, \forall
j=1,\cdots,N.
\end{equation}
The total upload rate, denoted as $u_j$, is constrained by the
uplink capacity. Hence, we also have $u_j = \sum_{i \neq j} r_{j
\rightarrow i} \leq U_j$ for all $j=1,\cdots,N$.

\begin{figure}
  \centering
  \includegraphics[width=0.4\textwidth]{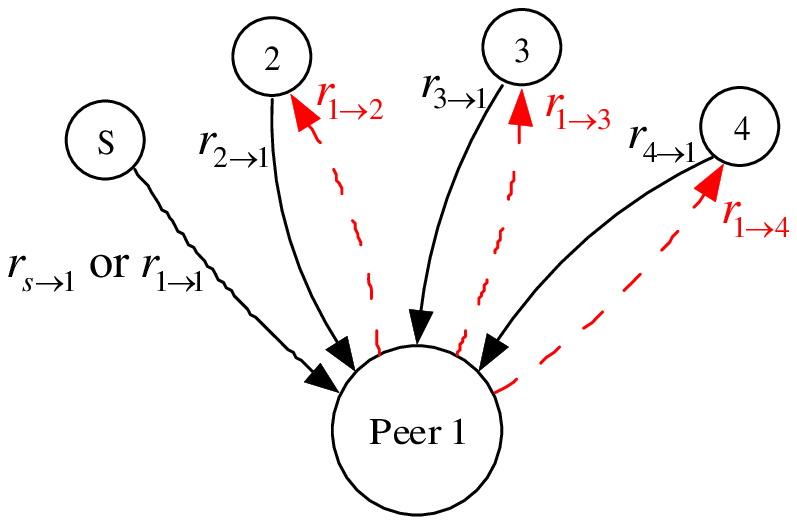}
  \caption{The peer model}\label{fig:peermodel}
\end{figure}

One example of the peer model is shown in Fig.~\ref{fig:peermodel}.
The downlink capacity and uplink capacity of peer 1 are $D_1$ and
$U_1$ respectively. Thus, the total download rate $r_{s\rightarrow
1} + \sum_{i=2}^{4} r_{i \rightarrow 1} = \sum_{i=1}^{4} r_{i
\rightarrow 1}$ has to be less than or equal to $D_1$, and the total
upload rate $\sum_{i=2}^{4} r_{1 \rightarrow i}$ has to be less than
or equal to $U_1$.

\subsection{The Time-Expanded Graph}
\label{sec:time-expanded-graph} As one of the key contributions of
\cite{Wu09}, Wu \emph{et al.} used a time-expanded graph to show how
the dynamic scenario decomposes into epochs.  This section applies
the time-expanded graph approach provided in \cite{Wu09} to the
\emph{static} case.

To obtain the time-expanded graph for a P2P network with $N$ peers,
we need to divide the time into $N$ \textit{epochs} according to the
finishing times of the peers. One peer finishes downloading at the
end of each epoch so that the number of epochs is always equal to
the number of peers. Let $\Delta t_i$ denote the duration of the
$i$-th epoch. Hence, $i$ receivers finish downloading by time $t_{i}
= \sum_{k=1}^{i} \Delta t_k$. If peers $i$ and $i+1$ finish
downloading at the same time,  $\Delta t_{i+1}=0$.

Each vertex in the original graph $G$ corresponds to $N$ vertices,
one for each epoch, in the time-expanded graph $G^{(N)}$ as follows:
We begin with the original P2P graph $G$ with node set $V = \{s,
1,\cdots ,N\}$ and allowed edge set $E$.  For each $v \in V$ and
each $n \in \{1,\cdots,N\}$, $G^{(N)}$ includes a vertex $v^{(n)}$
corresponding to the associated physical node $v$ in the $n$-th
epoch.  For each $e \in E$ going from $u$ to $v$ and each $n \in
\{1,\cdots,N\}$, $G^{(N)}$ includes an edge $e^{(n)}$ going from
$u^{(n)}$ to $v^{(n)}$ corresponding to the transmission from $u$ to
$v$ during the $n$-th epoch.

The subgraph $G^{(n)} = (V^{(n)}, E^{(n)})$ for $n=1,\cdots,N$
characterizes the network resource allocation in the $n$-th epoch.
To describe a rate allocation in the original graph $G$, edges are
typically labeled with the {\em{rate}} of information flow.
However, since each epoch in the time-expanded graph $G^{(N)}$ has a
specified duration, each of the $N$ edges in the time-expanded graph
corresponding to an edge in $G$ is labeled with the total amount of
information flow across the edge during its epoch.  This is the
product of the flow rate labeling that edge in the original graph
$G$ and the duration of the epoch.

The time-expanded graph also includes memory edges. For each $v \in
V$ and each $n \in \{1,\cdots,N-1\}$, there is an edge with infinite
capacity from $v^{(n)}$ to $v^{(n+1)}$. These memory edges reflect
the accumulation of received information by node $v$ over time.

As just described, the time-expanded graph not only describes the
network topology, but also characterizes the network resource
allocation over time until all peers finish downloading in a P2P
network.  As shown in \cite{Wu09} by Wu et al., even in the dynamic
scenario the network resource allocation can remain static
throughout each epoch without loss of optimality.  In this section,
we apply the time-expanded graph to the static scenario in which the
rate allocation remains fixed for the entire file transfer.

As an example, consider the following scenario.  A P2P network
contains a source node seeking to disseminate a file of unit size
($B=1$).  Its upload capacity is $U_S=2$. There are three peers
$\{1,2,3\}$ with upload capacities $U_1=U_2=U_3=1$ and download
capacities $D_1 = D_2 = D_3 = \infty$.

\begin{figure}
  \centering
  \includegraphics[width=0.4\textwidth]{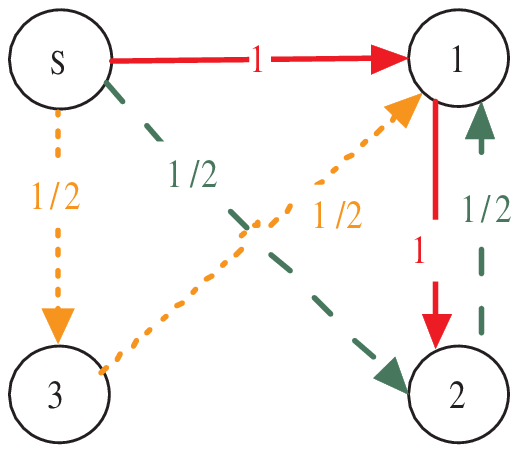}
  \caption{An example P2P graph $G$.  Edges are labeled with one possible rate allocation $r_{i \rightarrow j}$.}\label{fig:RateAllocationGraph}
\end{figure}

\begin{figure}
  \centering
  \includegraphics[width=0.5\textwidth]{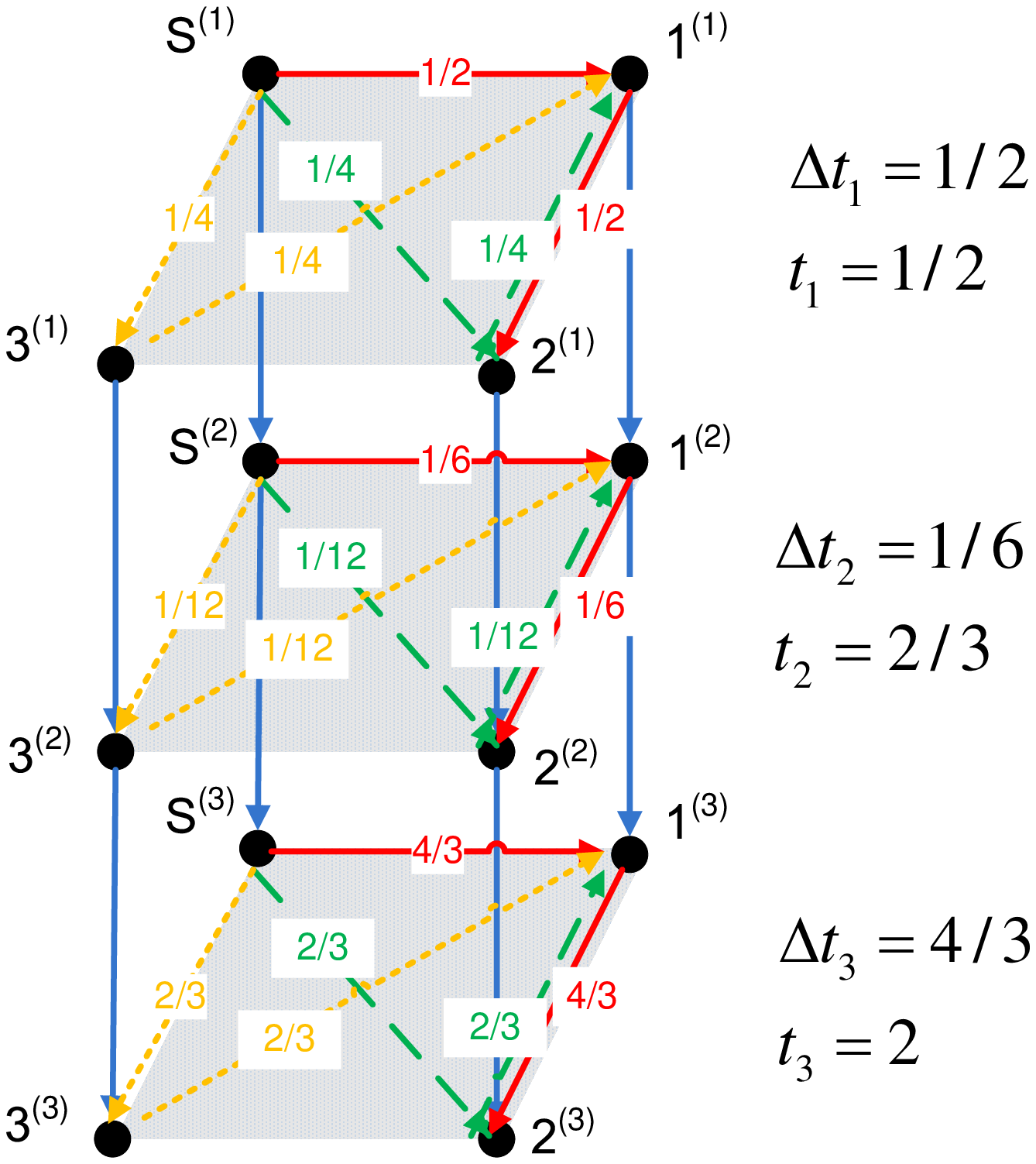}
  \caption{The time-expanded $G^{(3)}$ for the P2P graph $G$ shown in Fig.~\ref{fig:RateAllocationGraph}. Edges are labeled with the total information flow along the edge during the epoch.  This is the product of the rate allocation along the edge (from the graph in Fig.~\ref{fig:RateAllocationGraph}) and the duration of the epoch. Edges with zero flow are not shown.  }\label{fig:StaticTimeExpandedGraph}
\end{figure}

Fig.~\ref{fig:RateAllocationGraph} gives one possible static rate
allocation, showing the allocated rate for each edge of the original
P2P graph $G$.  (Edges with zero allocated rate are not shown.)  The
source node transmits with a rate of 1 to peer 1 and with rate 1/2
to peers 3 and 4.  Peer 1 transmits with rate 1 to peer 2 but does
not transmit to any other peers.  Peers 2 and 3 transmit with rate
1/2 to Peer 1, but do not transmit to any other peers.

Fig.~\ref{fig:StaticTimeExpandedGraph} shows the time-expanded graph
induced by the static rate allocation shown in
Fig.~\ref{fig:RateAllocationGraph}. Because there are three peers,
this time-expanded graph has 3 epochs. The peers are numbered in the
order they finish downloading; peer 1 finishes first followed by
peer 2 and then peer 3.  The first epoch lasts $\Delta t_1=1/2$ time
units, the second epoch lasts $\Delta t_2= 1/6$ time units, and the
third epoch lasts $\Delta t_3$ = 4/3 time units.

Peer 1 finishes first because it sees the full upload capacity of
the source.  As shown in Fig.~\ref{fig:RateAllocationGraph} it sees
rate 1 directly from the source.  The other half of the source
upload capacity is relayed to peer 1 by peers 2 and 3 immediately
after they receive it.  Hence peer 1 receives information with an
overall rate of $r_1=2$ and finishes downloading the entire file,
which has size $B=1$ at time $t_1=1/2$.  As a result, the duration
of the first epoch is $\Delta t_1=1/2$.

Peer 2 sees rate 1/2 directly from the source and rate 1 relayed to
peer 2 by peer 1.  Hence it sees an overall upload capacity of
$r_2=3/2$ and finishes downloading the entire file at time
$t_2=2/3$.  The duration of the second epoch can be computed as $t_2
- t_1 = 1/6$.

Because it receives no help from the other two peers, peer 3 sees an
overall upload rate of only $r_3 = 1/2$, which it receives directly
from the source. It finishes downloading the entire file at time
$t_3=2$.  The duration of the third epoch can be computed as $t_3 -
t_2 = 4/3$.

The sum of the download times for the example of
Figs.~\ref{fig:RateAllocationGraph}~and~\ref{fig:StaticTimeExpandedGraph}
is $1/2+2/3+2= 3~1/6$ .  Now let's consider an example that
minimizes the sum of the download times and in which peers finish at
the same time.

\begin{figure}
  \centering
  \includegraphics[width=0.4\textwidth]{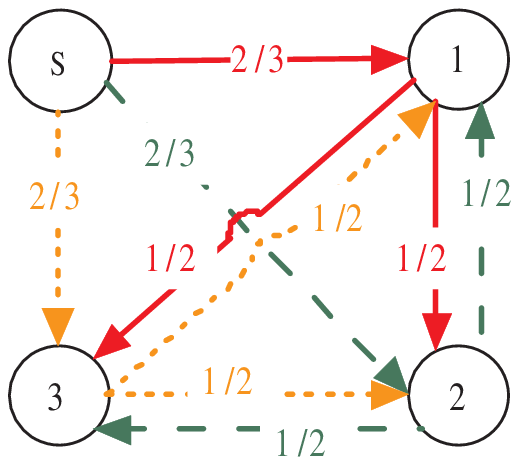}
  \caption{An example P2P graph $G$.  Edges are labeled with the rate allocation $r_{i \rightarrow j}$ that minimizes the sum of the download times.}\label{fig:OptRateAllocationGraph}
\end{figure}

\begin{figure}
  \centering
  \includegraphics[width=0.5\textwidth]{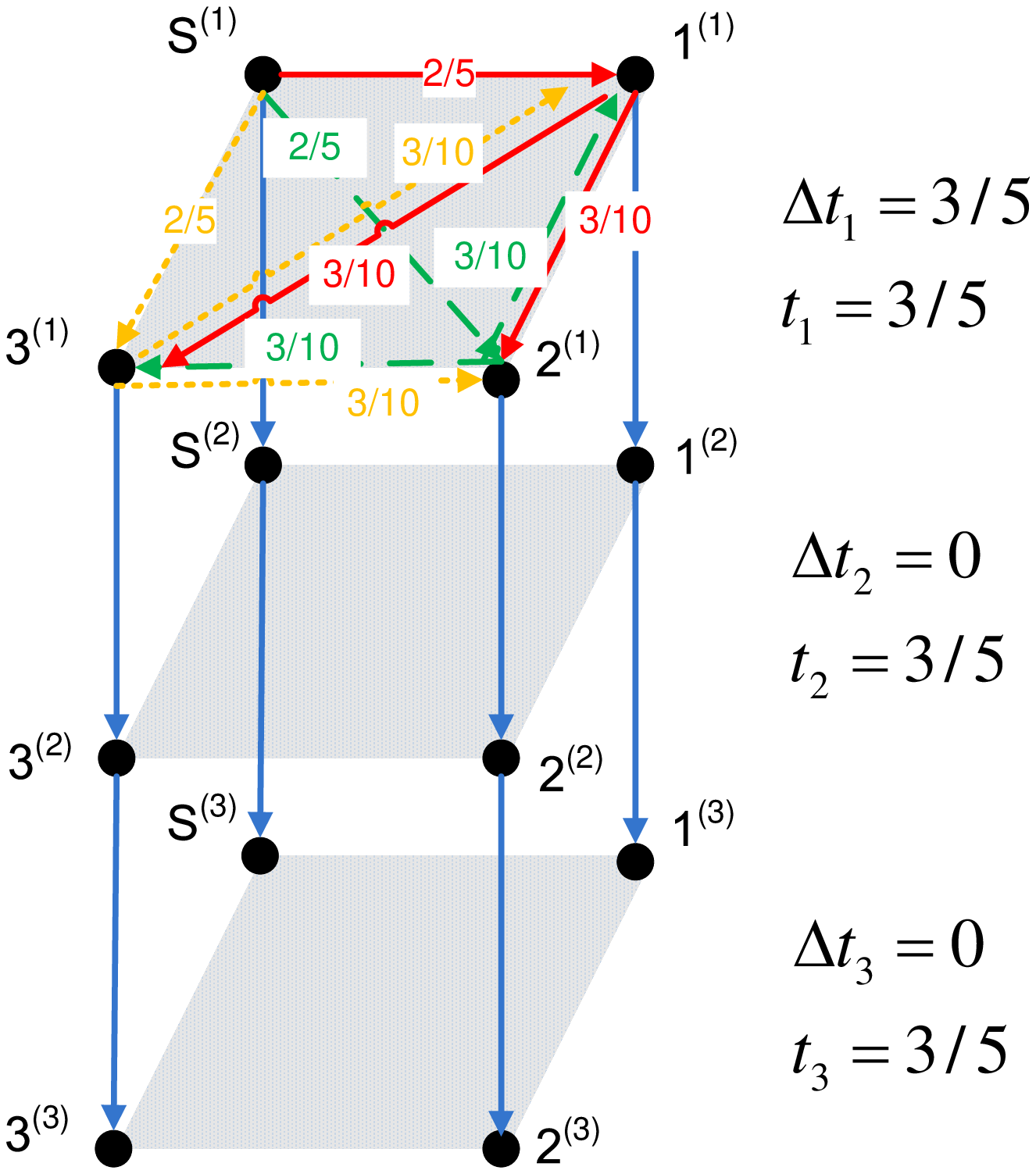}
  \caption{The time-expanded $G^{(3)}$ for the P2P graph $G$ shown in Fig.~\ref{fig:OptRateAllocationGraph}.  Edges are labeled with the total information flow along the edge during the epoch. This is the product of the rate allocation along the edge (from the graph in Fig.~\ref{fig:OptRateAllocationGraph})and the duration of the epoch. Edges with zero flow are not shown.}\label{fig:OptStaticTimeExpandedGraph}
\end{figure}

Fig.~\ref{fig:OptRateAllocationGraph} shows the rate allocation that
achieves the minimum possible sum of download times for a static
allocation in this scenario, which turns out to be 1~4/5. The
allocation shown in Fig.~\ref{fig:OptRateAllocationGraph} is
perfectly symmetric.  Each peer receives rate 2/3 directly from the
source and rate 1/2 from each of the two other peers.  Each peer
receives an overall rate of 5/3.  Hence all three peers finish
downloading simultaneously at $t=3/5$ and the second and third
epochs have zero duration.

\subsection{Transmission Flow Vectors and a Basic Network Coding Result}

In Section \ref{sec:time-expanded-graph} there was a tacit
assumption that all of the information received by a peer is {\em
useful}.  For example, we assumed that the information relayed from
peer 2 to peer 1 did not repeat information sent from the source to
peer 1.  In the examples of Section \ref{sec:time-expanded-graph},
one can quickly construct simple protocols that ensure that no
critical flows are redundant.  In this subsection, we review a
general result that uses network coding theory to show that there is
always a way to ensure that no critical flows are redundant.

Consider a general graph $G = (V,E)$, which could be either a
rate-allocation graph $G$ such as Figs.
\ref{fig:RateAllocationGraph} or \ref{fig:OptRateAllocationGraph} or
a time-expanded graph such as $G^{(3)}$ described in Figs.
\ref{fig:StaticTimeExpandedGraph} and
\ref{fig:OptStaticTimeExpandedGraph}. Denote $c(e)$ as the capacity
of the edge $e \in E$. A transmission flow from the source node $s$
to a destination node $i$ is a nonnegative vector
$\textit{\textbf{f}}$ of length $|E|$ satisfying the flow
conservation constraint: $\textrm{excess}_v (\textit{\textbf{f}}) =
0, \forall v \in V \backslash \{s, i\}$,where
\begin{equation}
\textrm{excess}_v (\textit{\textbf{f}}) = \sum_{e \in In(v)} f(e) -
\sum_{e \in Out(v)} f(e). \label{eqn:FlowConservation}
\end{equation}

The total flow supported by $\textit{\textbf{f}}$ is $\sum_{e \in
Out(s)} f(e)$.  This ``flow'' could be a flow {\em rate} with units
of bits per unit time if we are considering a rate allocation graph
such as Fig. \ref{fig:RateAllocationGraph} or it could be a total
flow with units of bits or packets or files if we are considering a
time-expanded graph such as Fig. \ref{fig:StaticTimeExpandedGraph}.

As an example, the flow vector $\textit{\textbf{f}}$ describing the
flow in Fig.~\ref{fig:StaticTimeExpandedGraph} from $S^{(1)}$ (the
source in the first epoch) to destination node $2^{(2)}$ (peer 2 in
the second epoch, when peer 2 finishes downloading) has the nonzero
elements $f(e)$ shown in Table \ref{tbl:FlowVector}.  Examining
Table \ref{tbl:FlowVector} verifies that the flow conservation
constraint (\ref{eqn:FlowConservation}) is satisfied and that the
total flow supported is equal to 1 file.

\begin{table}
\caption{Table showing nonzero elements $f(e)$ for the flow vector
$\textit{\textbf{f}}$ from $S^{(1)}$ to $2^{(2)}$ in
Fig.~\ref{fig:StaticTimeExpandedGraph}.} \label{tbl:FlowVector}
\begin{center}
\begin{tabular}{|c|l|}
\hline
$e$&$f(e)$\\
\hline \hline
$S^{(1)}\rightarrow 2^{(1)}$&$1/4$\\
\hline
$S^{(1)}\rightarrow 1^{(1)}$&$1/2$\\
\hline
$1^{(1)}\rightarrow 2^{(1)}$&$1/2$\\
\hline
$2^{(1)}\rightarrow 2^{(2)}$&$3/4$\\
\hline
$S^{(1)}\rightarrow S^{(2)}$&$1/4$\\
\hline
$S^{(2)}\rightarrow 2^{(2)}$&$1/12$\\
\hline
$S^{(2)}\rightarrow 1^{(2)}$&$1/6$\\
\hline
$1^{(2)}\rightarrow 2^{(2)}$&$1/6$\\
\hline
\end{tabular}
\end{center}
\end{table}

\newpage

The following lemma states that a given fixed flow (or flow rate)
can be achieved from the source to all destinations as long as there
is a feasible flow vector supporting the desired flow from the
source to each destination.  i.e. We can achieve this flow to all
destinations with network coding without worrying about possible
interactions of the various flows..

\newtheorem{NetworkCoding}{Lemma}
\begin{NetworkCoding}\label{Lemma:NetworkCoding}
\textbf{(Network Coding for Multicasting \cite{NetworkInfoFlow}
\cite{LinearNetworkCode})} In a directed graph $G = (V,E)$ with edge
capacity specified by a vector $\textit{\textbf{c}}$ , a multicast
session from the source node $s$ to a set of receivers $i \in
\{1,\ldots,N\}$ can achieve the same flow $r$ for each $i \in
\{1,\ldots,N\}$ if and only if there exits a set of flows
$\{\textit{\textbf{f}}_i\}$ such that
\begin{equation}
\textit{\textbf{c}} \geq \max _{i} \textit{\textbf{f}}_i
\label{eqn:networkcoding1} \end{equation} where
$\textit{\textbf{f}}_i$ is a flow from $s$ to $i$ with flow $r$.
Furthermore, if (\ref{eqn:networkcoding1}) holds, there exists a
linear network coding solution.
\end{NetworkCoding}

\subsection{A Convex Optimization}
Given an order in which the peers will finish downloading, say peer
$i$ finishes at the end of the $k_i$-th epoch, applying Lemma
\ref{Lemma:NetworkCoding} to the time-expanded graph $G^{(N)}$ with
the set of destination nodes $i \in \{1,\ldots,N\}$  gives a
characterization of all feasible downloading times, as concluded in
the following lemma.

\newtheorem{FeasibleDelay1}[NetworkCoding]{Lemma}
\begin{FeasibleDelay1}\label{Lemma:FeasibleDelay1}
\textbf{(Feasible Downloading Times with Given Order \cite{Wu09})}
Consider a P2P network in which node $D_i=\infty$. Given an order in
which the nodes will finish downloading a file with size $B$, say
node $i$ finishes at epoch $k_i$, a set of epoch durations ${\Delta
t_i}$ is feasible if and only if the following system of linear
inequalities has a feasible solution:
\begin{align}
\Delta t_i & \geq 0,\quad i=1,\cdots,N,\\
\textit{\textbf{g}} & \geq \textit{\textbf{f}}_i, \quad i=1,\cdots,N,\\
\sum_{v:u^{(i)}\rightarrow v^{(i)}} g_{u^{(i)}\rightarrow v^{(i)}} &
\leq c_u \Delta t_i, \quad \forall u \in V, i=1,\cdots,N,
\end{align}
where $c_u$ is the uplink capacity of peer $u$, and
$\textit{\textbf{f}}_i$ is a flow from first-epoch source node
$s^{(1)}$ to node $i$'s termination-epoch node $i^{(k_i)}$ with flow
rate $B$.
\end{FeasibleDelay1}

As an example, the epoch durations of
Fig.~\ref{fig:StaticTimeExpandedGraph} are feasible because each of
the flow vectors (one example was given in Table
\ref{tbl:FlowVector}) satisfy the feasibility constraints of Lemma
\ref{Lemma:FeasibleDelay1}.

Let $t_j$ denote the download time to peer $j$ for $j = 1,\cdots,N$.
Given a static network resource allocation $r_{i \rightarrow j}$,
$(i,j \in \{1,\cdots,N\})$ as shown for example in Fig.
~\ref{fig:RateAllocationGraph}, the maximum flow to peer $j$,
denoted as $r_j$, is equal to the minimum cut between source node
$s$ and peer $j$ in the rate-allocation graph (i.e. a graph such as
Fig.~\ref{fig:RateAllocationGraph}, not the time-expanded graph).
This follows from the Max-Flow-Min-Cut Theorem. Hence, $t_j \geq
\frac{B}{r_j}, \forall j$.

From applying network coding results such as Lemma
\ref{Lemma:NetworkCoding} to the rate allocation graph alone, we
cannot conclude much about feasible download times since Lemma
\ref{Lemma:NetworkCoding} addresses only the feasibility of the same
flow to all destinations.  However, by applying Lemma
\ref{Lemma:NetworkCoding} to the {\em time-expanded graph} we can
show that $t_j = \frac{B}{r_j}$ can be achieved simultaneously for
all $j=1,\cdots,N$. Lemma \ref{Lemma:StaticOptDelay} below states
this result.

\newtheorem{StaticOptDelay}[NetworkCoding]{Lemma}
\begin{StaticOptDelay}\label{Lemma:StaticOptDelay}
Given a static network resource allocation $r_{i \rightarrow j}$,
$(i,j=1,\cdots,N)$, for a P2P network, the only Pareto optimal
(smallest) download time vector is $t_j = \frac{B}{r_j}$ for
$j=1,\cdots,N$, where $r_j$ is the minimum cut from the source node
$s$ to peer $j$.
\end{StaticOptDelay}

\begin{proof}
It has been shown that $t_j \geq \frac{B}{r_j}$ for $j=1,\cdots,N$.
Hence, it is sufficient to show that $t_j = \frac{B}{r_j}$ for
$j=1,\cdots,N$ is achievable. Without loss of generality, assume
that $r_1 \geq r_2 \geq \cdots \geq r_N > 0$. Construct a static
scheme (i.e. a time-expanded graph $G^{(N)}$) as
follows:\\
(1) $\Delta t_i = \frac{B}{r_i}-\frac{B}{r_{i-1}}$, where $r_0 \triangleq \infty$;\\
(2) Flow capacity of edge $i^{(k)} \rightarrow j^{(k)}$ is $r_{i
\rightarrow j} \Delta t_k$ for $1 \leq i \neq j \leq N$ and
$k=1,\cdots, N$;\\
(3) Flow capacity of edge $s^{(k)} \rightarrow j^{(k)}$ is $r_{s
\rightarrow j} \Delta t_k$ for $j,k=1,\cdots, N$;\\
(4) Flow capacity of edge $j^{(k)} \rightarrow j^{(k+1)}$ is infinity for $j=1,\cdots,N$ and $k=1,\cdots, N-1$;\\
(5) The destination nodes in the time-expanded graph are node
$i^{(i)}$ for $i=1,\cdots,N$. In other words, peer $i$ finishes at
the end of $i$-th epoch. \\
According to the constructions (1) and (5), the download time to
peer $i$ is $t_i = \sum_{k=1}^{i} \Delta t_k = \frac{B}{r_i}$.
According to the constructions (2) and (3), in the subgraph
$G^{(k)}$, the maximum flow from $s^{(k)}$ to $i^{(k)}$ is equal to
$r_i \Delta t_k$ for all $i,k =1,\cdots,N$. Therefore, in this
time-expanded graph $G^{(N)}$, the maximum flow from source node $s$
to node $i^{(i)}$ is greater than or equal to
\[ \sum_{k=1}^{i} r_i \Delta t_k = B.\]
Therefore, by Lemma \ref{Lemma:NetworkCoding} and Lemma
\ref{Lemma:FeasibleDelay1}, there exists a linear network coding
solution to multicast a file with size $B$ from the source node $s$
to peer $i$ within download time $t_i = \frac{B}{r_i}$ for all
$i=1,\cdots, N$.
\end{proof}

The maximum flow $r_i$ can be found by solving a linear
optimization.  Specifically, a set of flow rates $\{r_i\}_{i=1}^{N}$
is feasible if and only if there exists a solution to the following
system of linear inequalities:
\begin{align}
r_{i \rightarrow j} & \geq 0, \quad \forall i,j=1,\cdots,N; \label{eq:feasiblerate1}\\
\sum_{i=1}^{N} r_{i \rightarrow i} & \leq U_s; \quad (\textrm{recall that } r_{i \rightarrow i} \triangleq r_{s \rightarrow i})\\
\sum_{j=1,j\neq i}^{N} r_{i \rightarrow j} & \leq U_i, \quad \forall
i=1,\cdots,N;\\
\sum_{j=1}^{N} r_{j \rightarrow i} & \leq D_i,\quad \forall
i=1,\cdots,N;
\end{align}
\begin{align}
0 \leq f_{i \rightarrow j}^{(k)} & \leq r_{i \rightarrow j},
\forall i,j,k = 1, \cdots, N;\\
f_{k \rightarrow j}^{(k)} & = 0, \quad \forall j \neq k;\\
 \sum_{j=1}^{N} f_{j \rightarrow i}^{(k)} & = \sum_{j=1,j\neq
i}^{N}
f_{i \rightarrow j}^{(k)}, \quad \forall i \neq k;\\
\sum_{i=1}^{N} f_{i \rightarrow k}^{(k)} & \geq r_k, \quad \forall
k=1,\cdots,N, \label{eq:feasiblerate8}
\end{align}
where $r_{i \rightarrow j}$ ($i,j =1,\cdots,N$) represents the
network resource allocation and $f_{i \rightarrow j}^{(k)}$ ($i,j
=1,\cdots,N$) is a flow from the source node $s$ to peer $k$.

By Lemma \ref{Lemma:StaticOptDelay}, the minimum WSDT is the
solution to the convex optimization of minimizing
$\sum_{i=1}^{N}W_iB/r_i$ subject to
(\ref{eq:feasiblerate1}-\ref{eq:feasiblerate8}). Thus, we can
conclude the following theorem:

\newtheorem{StaticPolynomial}{Theorem}
\begin{StaticPolynomial}\label{Theorem:StaticPolynomial}
Consider multicasting a file with size $B$ from a source node $s$ to
peers $\{1,\cdots,N\}$ in a P2P network with  both uplink and
downlink capacity limits. The minimum weighted sum downloading time
for the static scenario and the corresponding optimal static
allocation can be found in polynomial time by solving the convex
optimization of minimizing $\sum_{k=1}^{N}W_k B/r_i$ subject to the
constraints (\ref{eq:feasiblerate1}-\ref{eq:feasiblerate8}).
\end{StaticPolynomial}

Theorem \ref{Theorem:StaticPolynomial} gives a solution to the most
general static case that we are considering in this paper. However,
it can be extended further by adding other linear network
constraints (e.g. edge capacity constraints), which are not a
concern of this paper.

\subsection{The Uplink-Bandwidth-Sum Bound}
\label{sec:BoundingWSD}

For a P2P network with a source node and N peers, the convex
optimization in Theorem \ref{Theorem:StaticPolynomial} has
$N^3+N^2+N$ variables and $2N^3+3N^2+N+1$ linear constraints. The
complexity for the interior point method to solve this convex
optimization is $O((N^3)^{3.5})$ \cite{BookConvexOpt}.

Even though the convex optimization can be solved in polynomial
time, its complexity is still too high for practical applications
when $N$ is large. Hence, bounds on the minimum WSDT and static
schemes having network resource allocations that may be computed
with low complexity are desired. In this subsection, we provide an
analytical lower bound to the minimum WSDT with $O(N^2)$ complexity
for computing both the bound itself and the associated rate
allocations.

Consider the cut of $\{V \setminus i\} \rightarrow \{i\}$ for any
static allocation $r_{i \rightarrow j}$ $i,j\in \{1,\cdots, N\}$,
the maximum flow rate from the source node $s$ to peer $i$, $r_i$,
is limited by
\begin{equation}
r_i \leq \sum_{j=1}^{N} r_{j \rightarrow i} \leq
D_i,\label{eq:relaxconstraint1}
\end{equation}
and
\begin{align}
\sum_{i=1}^{N} r_i &  \leq \sum_{i=1}^{N}\sum_{j=1}^{N} r_{j \rightarrow i} \\
& = \sum_{j=1}^{N}  r_{j \rightarrow j} + \sum_{j=1}^{N}
\sum_{i=1,i \neq j}^{N} r_{j \rightarrow i} \\
& \leq U_s + \sum_{j=1}^{N}U_j. \label{eq:relaxconstraint2}
\end{align}
Consider the cut of $\{s\} \rightarrow \{1,\cdots,N\}$, $r_i$ is
also bounded by
\begin{equation}
r_i \leq \sum_{j=1}^{N} r_{j \rightarrow j} \leq U_s.
\label{eq:relaxconstraint3}
\end{equation}
Inequalities (\ref{eq:relaxconstraint1}) and
(\ref{eq:relaxconstraint3}) indicate that the downloading flow rate
for peer $i$ is limited by peer $i$'s downlink capacity and the
source node's uplink capacity respectively. These two constraints
are not only valid for the static scenario but also for dynamic
scenarios.

Inequality (\ref{eq:relaxconstraint2}) shows that the sum of the
downloading flow rates for all peers is bounded by the total amount
of the network uplink resource. Again, this constraint holds in both
the static and dynamic cases.

These three constraints characterize an outer bound to the region of
all feasible sets of $\{r_i\}_{i=1}^{N}$ satisfying
(\ref{eq:feasiblerate1} - \ref{eq:feasiblerate8}). Therefore, for
any static scheme, every set of feasible flow rates
$\{r_i\}_{i=1}^{N}$ must satisfy (\ref{eq:relaxconstraint1}),
(\ref{eq:relaxconstraint2}) and (\ref{eq:relaxconstraint3}).
However, not all $\{r_i\}_{i=1}^{N}$ satisfying
(\ref{eq:relaxconstraint1}), (\ref{eq:relaxconstraint2}) and
(\ref{eq:relaxconstraint3}) are feasible.

Consider the following example:  Let $B=1$,  $U_S=3$, and
$U_1=U_2=U_3=1$ (with $D_1 = D_2 = D_3 = \infty$), the downloading
flow rates $r_1=r_2=3,r_3=0$ satisfies the constraints
(\ref{eq:relaxconstraint1}), (\ref{eq:relaxconstraint2}) and
(\ref{eq:relaxconstraint3}), but are not feasible because there is
no solution to (\ref{eq:feasiblerate1} - \ref{eq:feasiblerate8})
with $r_1=r_2=3,r_3=0$, i.e., no static scenario to support
$r_1=r_2=3,r_3=0$ simultaneously.  Specifically, for $r_1+r_2=6$,
all upload capability must be deployed, including that of peer 3.
However, since $r_3=0$, any transmission by peer 3 would violate the
conservation-of-flow constraint.

Because all feasible sets of $\{r_i\}_{i=1}^{N}$ satisfy
(\ref{eq:relaxconstraint1}) (\ref{eq:relaxconstraint2}) and
(\ref{eq:relaxconstraint3}), the solution to the following
minimization problem provides a lower bound to the minimum WSDT for
the static scenario:
\begin{equation} \label{eq:minproblem}
  \begin{array}{cc}
    \min & \sum_{i=1}^{N} W_i \frac{B}{r_i} \\
    \textrm{subject to} & \sum_{i=1}^{N} r_i \leq U_s + \sum_{i=1}^{N}U_i \\
     & 0 \leq r_i \leq \tilde{D}_i \triangleq \min(D_i, U_s) , \forall i=1,\cdots,N, \\
  \end{array}
\end{equation}
where only $r_i$ $(i=1,\cdots,N)$ are the variables. Empirical
experiments presented in Section \ref{sec:staticsimulation} show
that this lower bound is tight for most P2P networks.

The minimization problem (\ref{eq:minproblem}) is a convex
optimization. Its optimal solutions are also the solutions to the
associated Karush$-$Kuhn$-$Tucker (KKT) conditions
\cite{BookConvexOpt}. The KKT conditions for problem
(\ref{eq:minproblem}) are
\begin{align}\label{eq:KKT}
-W_i \cdot \frac{1}{r_i^2} + \lambda + \mu_i & = 0, \quad
i=1,\cdots, N;\\
\sum_{i=1}^{N}r_i -U_s -\sum_{i=1}^{N}U_i & \leq 0, \quad \lambda
\geq 0;\\
r_i - \tilde{D}_i & \leq 0, \quad \mu_i \geq 0;\\
\lambda (\sum_{i=1}^{N}r_i -U_s -\sum_{i=1}^{N}U_i ) & = 0;\\
\mu_i (r_i -  \tilde{D}_i) & = 0, \quad i=1,\cdots, N.
\end{align}
Solving the KKT conditions yields the following optimal solution for
$\{r_i\}$:
\begin{eqnarray}
r^*_i = \Bigg \{ \begin{array}{cc}
                 \sqrt{W_i}\cdot R, & \textrm{if } \sqrt{W_i}\cdot R < \tilde{D}_i \\
                 \tilde{D}_i & \textrm{if } \sqrt{W_i}\cdot R \geq
                 \tilde{D}_i
               \end{array}~~,
              \label{eq:lowerboundsolution}
\end{eqnarray}
where $R$ is chosen such that
\begin{equation}
\sum_{i=1}^{N}r^*_i = \min(U_s+\sum_{i=1}^{N}U_i,
\sum_{i=1}^{N}\tilde{D}_i ). \label{eq:ratesumforbound}
\end{equation}
The lower bound to the WSDT for the static scenario is then
\begin{equation}\label{eq:lowerbound}
\sum_{i=1}^{N}W_i t_i \geq \sum_{i=1}^{N}W_i \frac{B}{r^*_i},
\end{equation}
with $r^*_i$ as specified in (\ref{eq:lowerboundsolution}).

For the special case where $W_i=1$ and $D_i = \infty$ $(i=1,\cdots,
N)$, the solution given in (\ref{eq:lowerboundsolution}) becomes
\begin{equation}
r^*_i = \min(U_s, \frac{U_s + \sum_{i=1}^{N}U_i}{N}),
\end{equation}
and the lower bound to the minimum WSDT is

\begin{equation}
\sum_{i=1}^{N}t_i \geq \frac{NB}{\min(U_s, \frac{U_s +
\sum_{i=1}^{N}U_i}{N})} .
\end{equation}

Mutualcast \cite{Mutualcast04} was designed to minimize the maximum
download time for the case where $D_i = \infty$.  However, since
Mutualcast can achieve the download time of $\frac{B}{\min(U_s,
\frac{U_s + \sum_{i=1}^{N}U_i}{N})}$ for all peers, it achieves the
lower bound of (\ref{eq:lowerbound}) for the $W_i=1$ case. This fact
shows both that the lower bound of (\ref{eq:lowerbound}) is tight
when  $W_i=1$ and $D_i = \infty$ and that Mutualcast minimizes sum
download time as well as the maximum download time when $D_i =
\infty$.

\section{Mutualcast and Extended Mutualcast for the Equal-Weight Static Case}
\label{sec:extended_mutualcast}

The concluding paragraph of Section \ref{sec:BoundingWSD} stated
that Mutualcast minimizes the sum download time for the case where
$D_i = \infty$.   In this section we extend Mutualcast to provide an
algorithm we call Extended Mutualcast that handles finite
constraints on $D_i$ (possibly delivering different rates to
different peers) while still minimizing the sum download time.

\subsection{Mutualcast}

Mutualcast delivers the same rate to every peer.  Assuming $D_i =
\infty$, Mutualcast can support peers with any rate $R \leq
\min(U_s, \frac{U_s + \sum_{i=1}^{N}U_i}{N})$.  The key aspect of
Mutualcast is that the source first delivers bandwidth to each node
according to how much that node can share with all other peers.
After that, if the source has any upload bandwidth left over, it is
divided evenly among all peers.  This leftover rate goes serves only
one peer; it is not relayed to any other peers.  Thus Mutualcast
first forms a series of depth-two trees from the source to all
nodes.  Then, if there is any source upload bandwidth left over, it
is used to form a series of depth-one trees.  Here is a
specification of the Mutualcast algorithm (without considering
helper nodes):
\begin{algorithm}
\caption{The Mutualcast Algorithm for Network Resource Allocation}
\label{alg:Mutualcast}
\begin{algorithmic}[1]
\STATE Given broadcast rate $R \leq \min(U_s, \frac{U_s +
\sum_{i=1}^{N}U_i}{N})$. \\
\STATE Given an ordering of the peers. (Without loss of generality,
assume
the order is $1,\cdots, N$.) \\
\FOR {$i=1$ to $N$} \STATE $r_{s \rightarrow i} \leftarrow \min (R, U_i/(N-1))$. \\
\STATE $r_{i \rightarrow j} \leftarrow r_{s \rightarrow i}$ for $j
\neq i$.\\
\STATE $R \leftarrow R - r_{s \rightarrow i}$.\\
\STATE $U_s \leftarrow U_s - r_{s \rightarrow i}$. \\
\ENDFOR \STATE $r_{s \rightarrow i} \leftarrow r_{s \rightarrow i} +
R$.
\end{algorithmic}
\end{algorithm}

Mutualcast delivers information to all peers at the same rate.  As
described in Algorithm \ref{alg:Mutualcast} the highest rate that
Mutualcast can deliver is
\begin{equation}
R = \min(U_s, \frac{U_s +\sum_{i=1}^{N}U_i}{N}).
\end{equation}
Consider two examples with ten peers, one in which $R=U_s$ and one
in which $R=\frac{U_s+\sum_{i=1}^{N}U_i}{N}$.

First is an example where $R=U_s$. Note that in general it is not
possible for any peer to receive information at a rate higher than
$U_s$.  Let $U_s=1$, $U_i=1$ for all ten peers, and $D_i = \infty$
for all ten peers.  Mutualcast achieves $R=U_s=1$ by having nine
peers receive rate 1/9 from the source and forward at that rate to
the nine other peers.  One peer receives no information directly
from the source because by the time the Mutualcast algorithm gets to
that peer, the source upload bandwith has been used up.

For an example where $R=\frac{U_s+\sum_{i=1}^{N}U_i}{N}$, a larger
$U_s$ is necessary.   Let $U_s=10$, $U_i=1$ for all ten peers and
$D_i = \infty$ for all ten peers.  Mutualcast achieves
$R=\frac{U_s+\sum_{i=1}^{N}U_i}{N}=2$.  In the first part of the
Mutualcast algorithm, all ten peers receive rate 1/9 from the source
and relay at that rate to the nine other peers.  At this point there
remains 80/9 of source upload bandwidth, which is distributed evenly
so that each peer receives a rate of 8/9 directly from the source
that it does not relay. In total, each peer receives rate 2 which is
comprised of rate 1 from other peers, rate 1/9 from the source that
it relays to the other peers, and rate 8/9 from the source that it
does not relay.

The basic Mutualcast algorithm does not consider download
constraints.  The slight modification of Mutualcast given below
includes download bandwidth constraints $D_i$ in the simplest
possible way.  Note that if all peers are to receive at the same
rate, that rate must be less than the smallest download constraint.
This is reflected in line 1 of Algorithm
\ref{alg:modifiedMutualcast}.

\begin{algorithm}
\caption{The Mutualcast Algorithm with Download Bandwidth
Constraints} \label{alg:modifiedMutualcast}
\begin{algorithmic}[1]
\STATE Given broadcast rate $R \leq \min\left(U_s, \frac{U_s +
\sum_{i=1}^{N}U_i}{N}, \min_{j \in \{1, \ldots, N\}} (D_j) \right)$. \\
\STATE Given an order of peers. (Without loss of generality, assume
the order is $1,\cdots, N$.) \\
\FOR {$i=1$ to $N$} \STATE $r_{s \rightarrow i} \leftarrow \min (R,
D_i, U_i/(N-1))$. \\
\STATE $r_{i \rightarrow j} \leftarrow r_{s \rightarrow i}$ for $j
\neq i$.\\
\STATE $R \leftarrow R - r_{s \rightarrow i}$.\\
\STATE $U_s \leftarrow U_s - r_{s \rightarrow i}$. \\
\STATE $D_j \leftarrow D_j - r_{s \rightarrow i}$ for
$j=1,\cdots,N$.\\
\ENDFOR \STATE $r_{s \rightarrow i} \leftarrow r_{s \rightarrow i} +
R$.
\end{algorithmic}
\end{algorithm}
As with the original Mutualcast, Algorithm
\ref{alg:modifiedMutualcast} delivers the same rate to every peer.
This alone is enough to prevent it from minimizes the sum download
time in general when there are download constraints.  However, it
will turn out to be an important component of Extended Mutualcast,
which is an algorithm that does minimize the sum download time under
general download constraints.

\subsection{Extended Mutualcast}

Setting $W_i=1$ for all $i$ in (\ref{eq:lowerboundsolution})
produces the following lower bound on the sum download time when
both upload and download constraints are considered:

\begin{equation}
\sum_{i=1}^{N}\frac{B}{r^*_i}, \label{eq:lowerboundMSDstart}
\end{equation}
where
\begin{align}
r^*_i &= \Bigg \{ \begin{array}{cc}
                 R, & \textrm{if } R < \tilde{D}_i \\
                 \tilde{D}_i & \textrm{if }  R \geq
                 \tilde{D}_i
               \end{array}\\
     &= \min (R,D_i,U_s),
\end{align}
where  $R$ is chosen such that
\begin{equation}
\sum_{i=1}^{N}r^*_i = \min(U_s+\sum_{i=1}^{N}U_i,
\sum_{i=1}^{N}\tilde{D}_i ). \label{eq:lowerboundMSDend}
\end{equation}
This lower bound can be achieved by a routing-based scheme that we
call Extended Mutualcast.

Consider a P2P network with constraints on peer uplink bandwidth and
peer downlink bandwidth.  Without loss of generality, assume that
$D_1 \leq \cdots \leq D_N$. Hence, $\tilde{D}_1 \leq \cdots \leq
\tilde{D}_N \leq U_s$ and $r^*_1 \leq \cdots \leq r^*_N$.  The
network resource allocation and the routing for Extended Mutualcast
are provided in Algorithms \ref{alg:extendedMutualcast} and
\ref{alg:extendedMutualcast2} respectively.

\begin{algorithm}
\caption{Network Resource Allocation for Extended Mutualcast}
\label{alg:extendedMutualcast}
\begin{algorithmic}[1]
\STATE Calculate $R$ and $r^*_i$ $(i=1, \cdots,N)$ from (\ref{eq:lowerboundMSDstart}-\ref{eq:lowerboundMSDend}). \\
\STATE Initialize network resource allocation $r_{i \rightarrow j}
\leftarrow 0$.\\
\IF {$R \leq \tilde{D}_1$}
\STATE $r^*_i = R$ for all $i=1, \cdots, N$. \\
\STATE Apply Algorithm \ref{alg:modifiedMutualcast} with rate $R$ to the network. \\
\ELSIF {$\tilde{D}_j < R \leq \tilde{D}_{j+1}$ for $j \in
\{1,\cdots,
N-1\}$} \STATE $r^*_i = \tilde{D}_i$ for $i \in \{1, \cdots, j\}$.\\
\STATE $r^*_i = R$ for $i \in \{ j+1, \cdots, N\}$.\\
\FOR {Step $i=1$ to $j$} \STATE Successively apply Algorithm
\ref{alg:modifiedMutualcast} with rate $\tilde{D}_i -
\tilde{D}_{i-1}$ ($\tilde{D}_0 \triangleq 0$) to the network with
the
source node $s$ and the ordered peers $\{i, \cdots, N\}$.  Note that with each successive application of Algorithm \ref{alg:modifiedMutualcast}, the values of $r_{i \rightarrow j}$ accumulate.\\
\ENDFOR \STATE Step $j+1$: Apply Algorithm
\ref{alg:modifiedMutualcast} one final time with rate $R -
\tilde{D}_j$ to the network with the source node $s$ and the ordered
peers
$\{j+1, \cdots, N\}$. Again, the values of $r_{i \rightarrow j}$ accumulate.\\
\ELSIF {$R \geq \tilde{D}_N$} \STATE $r^*_i = \tilde{D}_i = D_i$ for
$i=1,\cdots,N$.\\
\FOR {Step $i=1$ to $N$} \STATE Successively apply Algorithm
\ref{alg:modifiedMutualcast} with supporting rate $\tilde{D}_i -
\tilde{D}_{i-1}$ to the network with the
source node $s$ and the ordered peers $\{i, \cdots, N\}$.  Note that with each successive application of Algorithm \ref{alg:modifiedMutualcast}, the values of $r_{i \rightarrow j}$ accumulate.\\
\ENDFOR \ENDIF
\end{algorithmic}
\end{algorithm}

\begin{algorithm}
\caption{Routing Scheme for Extended Mutualcast}
\label{alg:extendedMutualcast2}
\begin{algorithmic}[1]
\STATE Given $R$ and $r^*_i$ $(i=1, \cdots,N)$ from (\ref{eq:lowerboundMSDstart}-\ref{eq:lowerboundMSDend}). \\
\STATE Given the network resource allocation $r_{i \rightarrow j}$
$(i,j =1,\cdots,N)$ by Algorithm \ref{alg:extendedMutualcast} where
$r_{i \rightarrow i} \triangleq r_{s \rightarrow i} \geq r_{i
\rightarrow
j}$. (This routing scheme is based on that network resource allocation.)\\
\STATE Partition the whole file into many chunks.\\
\IF {$R \leq \tilde{D}_1$} \STATE Apply the routing scheme of
Mutualcast \cite{Mutualcast04}. That is, for each $i=1,\cdots, N$
and some $j \neq i$, accumulatively route $\frac{r_{i \rightarrow
j}}{R}$ fraction of all chunks from the source node to peer $i$, and
then copy and route them from peer $i$ to other peers.
Accumulatively route the rest of
the chunks are from the source node to all peers directly.\\
\ELSIF {$\tilde{D}_j < R \leq \tilde{D}_{j+1}$ for $j=1,\cdots,
N-1$.} \STATE For $k =1, \cdots, j$, broadcast
$\frac{\tilde{D}_k-\tilde{D}_{k-1}}{R}$ fraction of all chunks to
peers $\{k,\cdots,N\}$ by Mutualcast.  Broadcast the rest of the
chunks to peers $\{j+1,\cdots,N\}$ by Mutualcast.
\STATE Until peers $\{j+1,\cdots, N\}$ finish downloading.\\
\FOR {Step $i=j$ to $1$}
\STATE  In Step $i$, the interesting chunks are those peer $i$ hasn't received.\\
\STATE For $k =1, \cdots, i$, accumulatively broadcast
$\frac{\tilde{D}_k - \tilde{D}_{k-1}}{\tilde{D}_i}$ fraction of the
interesting chunks to peers $k,\cdots,i$ by Mutualcast.
\STATE Until peer $i$ finishes downloading.\\
\STATE Note that peers $i+1, \cdots, N$ finish downloading before
Step $i$.\\
\STATE Note that prior to Step $i$, none of the peers $1, \cdots ,
i$ contain the interesting chunks broadcast during Step $i$. \ENDFOR
\ELSIF {$R \geq \tilde{D}_N$} \FOR {Step $i=N$ to $1$} \STATE In
Step $i$, the interesting chunks are those
peer $i$ hasn't received.\\
\STATE For $k =1, \cdots, i$, accumulatively broadcast
$\frac{\tilde{D}_k - \tilde{D}_{k-1}}{\tilde{D}_i}$ fraction of the
interesting chunks to peers $k,\cdots,i$ by Mutualcast.
\STATE Until peer $i$ finishes downloading.\\
\ENDFOR \ENDIF
\end{algorithmic}
\end{algorithm}

\clearpage

The network resource allocation for Extended Mutualcast (Algorithm
\ref{alg:extendedMutualcast}) is obtained by successively applying
Algorithm \ref{alg:modifiedMutualcast} to the P2P network or part of
the P2P network.  The network resource allocation by Algorithm
\ref{alg:extendedMutualcast} has $r_{s \rightarrow i} \geq r_{i
\rightarrow j}$ for all $i,j$. The flow rate to peer $i$, $r_i$, is
then equal to its download rate $\sum_{j=1}^{N}r_{j \rightarrow i}$.
The routing scheme for Extended Mutualcast (Algorithm
\ref{alg:extendedMutualcast2}) guarantees that the entire flow rate
$r_i$ is {\em useful}.  For the Extended Mutualcast rate allocation,
$r_i=\min(R,\tilde{D}_i)$ so that the lower bound
(\ref{eq:lowerboundMSDstart}-\ref{eq:lowerboundMSDend}) on sum
download time is achieved.  Theorem \ref{Theorem:SumDelay} formally
states and proves this fact.

\newtheorem{SumDelay}[StaticPolynomial]{Theorem}
\begin{SumDelay}\label{Theorem:SumDelay}
\textbf{(Minimum Sum Download time)} Consider multicasting a file
with size $B$ from a source node $s$ to peers $\{1,\cdots,N\}$ in a
P2P network with constraints on peer uplink bandwidth and peer
downlink bandwidth. The minimum sum download time for the static
scenario is $\sum_{i=1}^{N}\frac{B}{r^*_i}$, where $r^*_i$, the flow
rate to peer $i$, follows from
(\ref{eq:lowerboundMSDstart}-\ref{eq:lowerboundMSDend}).
\end{SumDelay}
\begin{proof}
(\textbf{Converse}) From
(\ref{eq:lowerboundMSDstart}-\ref{eq:lowerboundMSDend}),
$\sum_{i=1}^{N}\frac{B}{r^*_i}$ is a lower bound on the minimum sum
download time. Hence, any sum download time less than
$\sum_{i=1}^{N}\frac{B}{r^*_i}$
is not achievable. \\
(\textbf{Achievability}) It is sufficient to show that (a) Extended
Mutualcast is applicable to any P2P network, and (b) Extended
Mutualcast provides a static scenario in which the flow rate from
the source node to peer $i$ is $r^*_i$ of
(\ref{eq:lowerboundMSDstart}-\ref{eq:lowerboundMSDend}).\\
(\textbf{To Show (a)}) It is sufficient to show that in Algorithm
\ref{alg:extendedMutualcast}, the rate for each applied Algorithms
\ref{alg:modifiedMutualcast} is attainable. In other words, each
rate for the applied network is less than or equal to the minimum of
the source node's uplink capacity
and the total uplink resource over all of the peers. \\
\begin{itemize}
\item If $R \leq \tilde{D}_1$, then $R \leq U_s$ and $R \leq \frac{U_s
+ \sum_{i=1}^{N}U_i}{N}$. Hence, the rate $R$ is attainable for
Algorithm \ref{alg:modifiedMutualcast} in Line 4, Algorithm
\ref{alg:extendedMutualcast}.
\item If $\tilde{D}_j < R \leq \tilde{D}_{j+1}$, consider the worst
case of $D_i = R$ for $i=j+1,\cdots, N$ and $U_s = R$. In this case,
we have
\begin{align}
r^*_i & = \tilde{D}_i, \quad i=1,\cdots,N; \\
U_i & \leq D_i = \tilde{D}_i, \quad i=1,\cdots, j; \\
\tilde{D}_1 \leq \cdots \leq \tilde{D}_j & < R = \tilde{D}_{j+1}
= \cdots = \tilde{D}_N = U_s; \label{eq:3}\\
\sum_{i=1}^{N} \tilde{D}_i & = U_s + \sum_{i=1}^{N}U_i. \label{eq:2}
\end{align}
Denote $U^{(i)}_p$ as the total amount of the peers' uplink resource
used after Step $i$, and $U^{(i)}_s$ as the total amount of the
source node's uplink resource used after Step $i$. For Step 1,
$\tilde{D}_1 \leq U_s$ and $\tilde{D}_1 \leq \frac{\sum_{i=1}^{N}
\tilde{D}_i}{N} = \frac{U_s + \sum_{i=1}^{N}U_i}{N}$. Hence,
Algorithm \ref{alg:modifiedMutualcast} in Step 1 is feasible.
Suppose Algorithm \ref{alg:modifiedMutualcast} is feasible for Step
1 to Step $n$ ($1\leq n \leq j$). Then $U^{(i)}_s = \tilde{D}_i$ and
$U^{(i)}_p = \sum_{k=1}^{i} \tilde{D}_k + (N-i-1)\tilde{D}_i$.
Hence,
\begin{equation}
U^{(i)}_p  \geq \sum_{k=1}^{i} \tilde{D}_k  = \sum_{k=1}^{i} D_k
\geq \sum_{k=1}^{i} U_k,
\end{equation}
which indicates that Algorithm \ref{alg:modifiedMutualcast} for Step
1 to Step $i$ fully deploys the uplink resources of peers
$1,\cdots,i$.

Now consider Algorithm \ref{alg:modifiedMutualcast} for Step $n+1$,
the supporting rate is $\tilde{D}_{n+1} - \tilde{D}_n$. The source
node's uplink is $R-\tilde{D}_n$. The total uplink resource is
\begin{align}
& U_s + \sum_{i=1}^{N}U_i - (U^{(n)}_s + U^{(n)}_p) \label{eq:AvailableUplink}\\
 = & \sum_{i=1}^{N}\tilde{D}_i - (\sum_{k=1}^{n}\tilde{D}_k +
 (N-n)\tilde{D}_n) \label{eq:1}\\
 = & \sum_{k=n+1}^{N} (\tilde{D}_i - \tilde{D}_n) \\
 \geq & (N-n)(\tilde{D}_{n+1} - \tilde{D}_n), \label{eq:4}
\end{align}
where (\ref{eq:1}) follows from (\ref{eq:2}), and (\ref{eq:4})
follows from (\ref{eq:3}). Hence, the rate $\tilde{D}_{n+1} -
\tilde{D}_n$ is less than or equal to the total available uplink
resource (\ref{eq:AvailableUplink}) divided by the number of peers,
$N-n$. We also can see that $\tilde{D}_{n+1} - \tilde{D}_n$ is less
than or equal to the available source node's uplink bandwidth,
$R-\tilde{D}_n$. Therefore, Algorithm \ref{alg:modifiedMutualcast}
for Step $n+1$ is also feasible. By induction, Algorithm
\ref{alg:modifiedMutualcast} is feasible for every step.
\item If $R \geq \tilde{D}_{N}$, then
\begin{align}
D_1 \leq \cdots & \leq D_N  \leq R \leq U_s; \\
r^*_i & = D_i = \tilde{D}_i, \quad i=1,\cdots,N; \\
\sum_{i=1}^{N}r^*_i & = \sum_{i=1}^{N}D_i \leq \sum_{i=1}^{N}U_i
+U_s.
\end{align}
Consider the worst case of $\sum_{i=1}^{N}D_i = \sum_{i=1}^{N}U_i
+U_s$. For this worst case, Algorithm \ref{alg:modifiedMutualcast}
in Line 14 is feasible following an argument similar to that for the
case of $\tilde{D}_j < R \leq \tilde{D}_{j+1}$.
\end{itemize}
Therefore, Extended Mutualcast in Algorithm
\ref{alg:extendedMutualcast} is applicable to any P2P network.\\
(\textbf{To Show (b)}) From Algorithms \ref{alg:modifiedMutualcast}
and \ref{alg:extendedMutualcast}, Extended Mutualcast constructs a
static scenario with $r_{s \rightarrow i} \triangleq r_{i
\rightarrow i} \geq r_{i \rightarrow j}$ for $i,j=1,\cdots,N$, and
$\sum_{j=1}^{N}r_{j \rightarrow i}  \geq \min (R , \tilde{D}_i) =
r^*_i$. Hence, the maximum flow from the source node to peer $i$ is
larger than or equal to
\begin{align}
& \sum_{j=1,j\neq i}^{N} \min (r_{s \rightarrow j}, r_{j \rightarrow
i}) + r_{s \rightarrow i} \\
= & \sum_{j=1,j\neq i}^{N} r_{j \rightarrow i} + r_{s \rightarrow
i} \\
\geq & r^*_i.
\end{align}
Therefore, Extended Mutualcast provides a static scenario in which
the flow rate from the source node to peer $i$ is $r^*_i$ of
(\ref{eq:lowerboundMSDstart}-\ref{eq:lowerboundMSDend}).
\end{proof}

Theorem \ref{Theorem:SumDelay} showed that Extended Mutualcast
minimizes the sum download time for any static P2P network.  When
the total uplink bandwidth resource is sufficiently abundant,
Extended Mutualcast also minimizes the weighted sum download time
for any set of weights because all peers are downloading at their
limit of $\tilde{D}_i$.  Corollary \ref{Corollary:EnoughUplink}
formally states and proves this fact.
\newtheorem{EnoughUplink}{Corollary}
\begin{EnoughUplink}\label{Corollary:EnoughUplink}
Consider multicasting a file with size $B$ from a source node $s$ to
peers $\{1,\cdots,N\}$ in a P2P network with constraints on peer
uplink bandwidth and peer downlink bandwidth. If $U_s +
\sum_{i=1}^{N}U_i \geq \sum_{i=1}^{N}\tilde{D}_i$, the set of the
flow rates $r_i = \tilde{D}_i$ ($i=1,\cdots,N$) is attainable.
Hence, the minimum weighted sum download time for the static
scenario is $\sum_{i=1}^{N}W_i\frac{B}{\tilde{D}_i}$ for any given
weights $W_i$.
\end{EnoughUplink}

\begin{proof}
(\textbf{Achievability}) Note that when $U_s + \sum_{i=1}^{N}U_i
\geq \sum_{i=1}^{N}\tilde{D}_i$, $r^*_i$ of
(\ref{eq:lowerboundMSDstart}-\ref{eq:lowerboundMSDend}) is equal to
$\tilde{D}_i$. By Theorem \ref{Theorem:SumDelay}, Extended
Mutualcast can achieve the download rates $r^*_i =
\tilde{D}_i$.\\
(\textbf{Converse}) By Max-Flow Min-Cut Theorem, the maximum flow
from source node to peer $i$ is limited by $\tilde{D}_i = \min (D_i,
U_s)$. Hence, any weighted sum downloading time less than
$\sum_{i=1}^{N}W_i\frac{B}{\tilde{D}_i}$ is not achievable.
\end{proof}

\section{A Depth-2 Approach for the Minimizing Weighted Sum Download Time}
\label{sec:depth2}

Section \ref{sec:extended_mutualcast} provided a complete solution
(Extended Mutualcast) for achieving the minimum sum download time
with constraints on both peer uplink bandwidth and peer downlink
bandwidth.  That section concluded by showing that if the total
uplink resource is sufficiently abundant, Extended Mutualcast
minimizes WSDT for any set of weights.  This section attacks the
minimization of WSDT more broadly.

Mutualcast and Extended Mutualcast construct only two types of trees
to distribute content. The first type is a depth-1 tree as shown in
Fig.~\ref{fig:typetree}(a). The source node $s$ broadcasts content
to all peers directly with rate $r^{(1)}_{s \rightarrow i}$,
$i=1,\cdots,N$.    The second type is a depth-2 tree as shown in
Fig.~\ref{fig:typetree}(b). The source node distributes content to
peer $i$ with rate $r^{(2)}_{s \rightarrow i}$, and then peer $i$
relays this content to all other peers.

\begin{figure}
  \centering
  \includegraphics[width=0.6\textwidth]{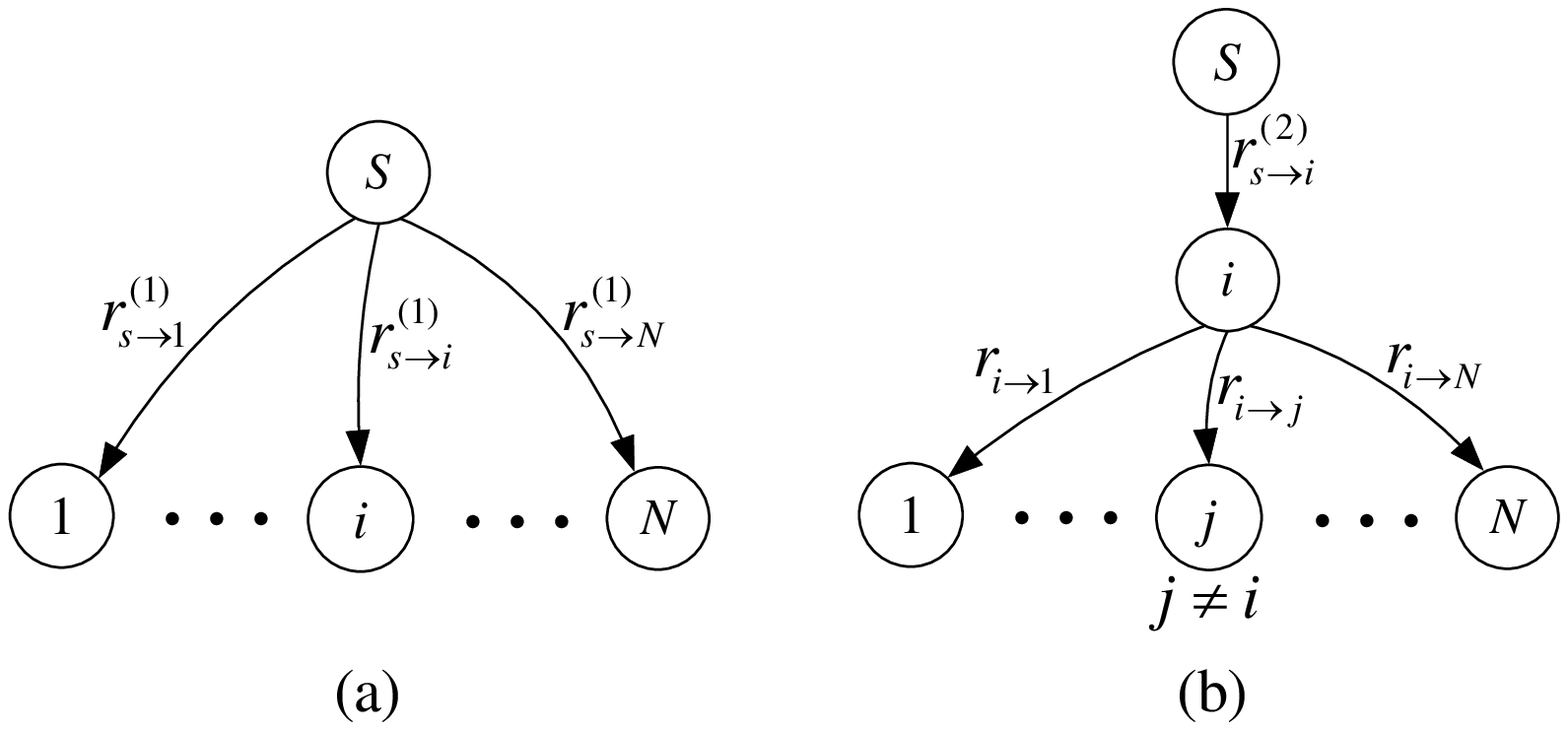}
  \caption{The two tree sturctures used by Mutualcast and Extended Mutualcast: (a) Depth-1 tree; (b) Depth-2 tree.}\label{fig:typetree}
\end{figure}

In Mutualcast, the rates $r^{(1)}_{s \rightarrow i}$ are constrained
to be equal for all $i$.  Also, for a fixed $i$, $r^{(2)}_{s
\rightarrow i}=r_{i \rightarrow j}$ for all $j\in \{1,\cdots,
N\};~~j\neq i$.  These constraints on the network resource
allocation simplify the mechanism design and allow a simple
routing-based scheme.  These two constraints together ensure that
each peer downloads content at the same rate.  However, to optimize
WSDT peers surely need to download content and different rates.

In Section \ref{sec:extended_mutualcast} we saw that peers needed to
download content at different rates to minimize the sum download
time with peer downlink bandwidth constraints.  The Extended
Mutualcast algorithm provided a way to serve the peers at different
rates corresponding to their download bandwidth constraints so as to
minimize the sum download time.  However, Extended Mutualcast
required successive applications of Mutualcast which led to a
complicated routing protocol.

In order to serve peers at different rates to minimize WSDT and
still maintain a simple mechanism design, we apply the technique of
rateless coding at the source node. A rateless code is an erasure
correcting code.  It is rateless in the sense that the number of
encoded packets that can be generated from the source message is
potentially limitless \cite{Fountain05}.  Suppose the original file
size is $B$ packets, once the receiver has received any $B'$
packets, where $B'$ is just slightly greater than $B$, the whole
file can be recovered.

Fountain codes \cite{Fountain05}, LT codes \cite{LT02}, and raptor
codes \cite{Raptor03} are rateless erasure codes. LT codes have
linear encoding complexity and sub-linear decoding complexity.
Raptor codes have linear encoding and decoding complexities. The
percentage of the overhead packets goes to zero as $B$ goes to
infinity. In practice, the overhead is about 5\% for LT codes with
file size $B \simeq 10000$ \cite{Fountain05}.  This sub-section
focuses on applying rateless erasure codes for P2P file transfer
instead of designing rateless erasure codes. Hence, we assume the
overhead of the applied rateless erasure code is zero for
simplicity.  We note that if redundancy does not need to be
limitless, there are solutions that provide zero overhead
\cite{CourtadeITA10}.

\subsection{The Rateless-Coding-Based Scheme}

We propose a rateless-coding-based scheme that constructs the two
types of trees in Fig.~\ref{fig:typetree} to distribute the content
as did Mutualcast and Extended Mutualcast. The source node first
partitions the whole file into $B$ chunks and applies a rateless
erasure code to these $B$ chunks producing a potentially limitless
number of chunks.

For the depth-1 tree, the source node broadcasts different
rateless-coded chunks directly to each peer. For the depth-2 trees,
The source node sends different rateless-coded chunks to each peer,
and then that peer relays some or all of those chunks to other some
or all of the other peers.   A key point is that every chunk
transmitted by the source is different from every other chunk
transmitted by the source.  This condition guarantees that all
chunks received by a peer are useful (because they are not a
repetition of a previously received chunk). Hence, a peer can decode
the whole file as long as it receives $B$ coded chunks.

The rateless-coding-based scheme allows peers to download content at
different rates with a simpler mechanism than the routing-based
approach of Extended Mutualcast.  Peers don't have to receive
exactly the same chunks to decode the whole file.  Hence, the two
types of tree structures can be combined as one tree structure with
depth 2, but without the constraint that the rate from the peer to
its neighbors has to equal the rate from the source to the peer.

The source node sends coded chunks to peer $i$ with rate $r_{s
\rightarrow i} = r^{(1)}_{s \rightarrow i} + r^{(2)}_{s \rightarrow
i}$, and peer $i$ relays some of them to peer $j$ ($j \neq i$) with
rate $r_{i \rightarrow j} \leq  r_{s \rightarrow i}$.   Note that
the values of $r_{i \rightarrow j}$ do not even need to be the same
for a fixed value of $i$ and different values of $j$.

Another benefit of applying a rateless coding approach is that it is
robust to packet loss in the Internet if we allow some extra rate
for each user.

Assuming rateless coding at the source node and constraining the P2P
network to include only depth-2 trees as discussed above, the
network resource allocation that minimizes WSDT can be obtained by
solving the following convex optimization problem.
\begin{equation} \label{eq:ratelessminproblem}
  \begin{array}{cc}
    \min & \sum_{i=1}^{N} W_i \frac{B}{r_i} \\
    \textrm{subject to} & 0 \leq r_{i \rightarrow j} \leq r_{i \rightarrow i}, \forall i,j=1,\cdots,N, \\
     &  \sum_{i=1}^{N} r_{i \rightarrow i} \leq U_s,\\
     & \sum_{j=1,j\neq i}^{N}r_{i \rightarrow j} \leq U_i, \forall
     i=1,\cdots,N,\\
     & r_i = \sum_{j=1}^{N}r_{j \rightarrow i}
     \leq D_i,\forall i=1,\cdots,N,\\
  \end{array}
\end{equation}
where $r_{i \rightarrow i} \triangleq r_{s \rightarrow i}$. The
complexity for the interior point method to solve this convex
optimization is $O((N^2)^{3.5})$ \cite{BookConvexOpt}.

For the case of $W_i=1, D_i = \infty$, the optimal resource
allocation is, of course, the same as that of Mutualcast.

For the case of $W_i =1$ and finite values of $D_i$, Algorithm
\ref{alg:extendedMutualcast} provides an optimal network resource
allocation that certainly also solves (\ref{eq:ratelessminproblem}).
A key point is that the routing of Algorithm
\ref{alg:extendedMutualcast2} becomes unnecessary if the source
employs rateless coding.  Peers need only relay the appropriate
number of chunks to the appropriate neighbors without worrying about
{\em which} chunks are relayed.

For other cases, we provide a network resource allocation that we
have not proven to be optimal.  We will see in Section
\ref{sec:staticsimulation} that its performance achieves the lower
bound (\ref{eq:lowerbound}) across a wide range of
parameterizations.

\subsection{Resource Allocation for Networks with $D_i=\infty$}

Consider a P2P network in which peer uplink bandwidth is constrained
but $D_i = \infty$ for $i=1,\cdots, N$. If $\sum_{i=1}^{N}U_i \geq
(N-1)U_s$, then the resource allocation of $r_{i \rightarrow j} =
\frac{U_s U_i}{\sum_{i=1}^{N}U_i}$ achieves the minimum WSDT with
flow rates $r_i = U_s$ for all $i=1,\cdots,N$.  (This is the case
discussed at the end of Section \ref{sec:extended_mutualcast}.)
Otherwise, consider the following water-filling-type solution:
\begin{equation}
\tilde{r}_i = \Bigg \{ \begin{array}{cc}
                 \sqrt{W_i}\cdot R, & \textrm{if } \sqrt{W_i}\cdot R < U_s, \\
                 U_s & \textrm{if } \sqrt{W_i}\cdot R \geq U_s,
               \end{array}
              \label{eq:waterfilling2}
\end{equation}
where $R$ is chosen such that
\begin{equation}
\sum_{i=1}^{N}\tilde{r}_i = U_s + \sum_{i=1}^{N}U_i - \max_{k}
(\tilde{r}_k). \label{eq:ratelessWF}
\end{equation}

The potential suboptimality of this approach comes from the
subtraction of $\max_{k} (\tilde{r}_k)$ on the right side of
(\ref{eq:ratelessWF}) which does not appear in
(\ref{eq:ratesumforbound}). Note that when $\max (\tilde{r}_k) \ll
U_s + \sum_{i=1}^{N}U_i$ (this is true for large $N$), $\tilde{r}_i$
is close to $r^*_i$ corresponding to the lower bound
(\ref{eq:lowerboundsolution}).

We now show that the proposed suboptimal network resource allocation
ensures that the flow rate to peer $i$ is larger than or equal to
$\tilde{r}_i$ of (\ref{eq:waterfilling2}). Hence, the WSDT for the
proposed suboptimal resource allocation is very close to the lower
bound to the minimum WSDT for large networks.

First assign the rates for the depth-2 trees with
\begin{equation}
r^{(2)}_{s \rightarrow i} = c \frac{U_i \max
(\tilde{r}_k)}{\sum_{k=1}^{N}\tilde{r}_k -
\tilde{r}_i}\label{eq:rateassign1}\end{equation} and
\begin{equation}
r_{i \rightarrow j} = c \frac{U_i
\tilde{r}_j}{\sum_{k=1}^{N}\tilde{r}_k - \tilde{r}_i},
\label{eq:rateassign2}\end{equation} where $c$ is chosen to be the
largest possible value satisfying
\begin{align}
\sum_{i=1}^{N}r^{(2)}_{s \rightarrow i} & \leq U_s \label{eq:c_constraint1}\\
\sum_{j=1,j\neq i}^{N}r_{i \rightarrow j} & \leq U_i.
\label{eq:c_cosntraint2}
\end{align}
Plugging (\ref{eq:rateassign1}) (\ref{eq:rateassign2}) into
(\ref{eq:c_constraint1}) (\ref{eq:c_cosntraint2}), and obtain
\begin{equation}\label{eq:c}
c = \min (1, \frac{U_s}{\max(\tilde{r}_k)\alpha} ),
\end{equation}
where $\alpha = \sum_{i=1}^{N}\frac{U_i}{\sum_{k=1}^{N}\tilde{r}_k -
\tilde{r}_i}$.

If $c = \frac{U_s }{ \alpha\max(\tilde{r}_k)}$, then the depth-2
trees have already fully deployed the source node's uplink. The rate
assignment for depth-2 trees is the network resource allocation for
the rateless-coding-based scheme.

If $c = 1$, then the depth-2 trees have fully deployed all peers'
uplinks, but not the source node's uplink. Hence, we can further
deploy the rest of the source node's uplink to construct the depth-1
tree. After constructing the depth-2 trees, the flow rate to peer
$i$ is
\[ \beta_i \triangleq
r^{(2)}_{s \rightarrow i} + \sum_{j=1,j\neq i}^{N}r_{j \rightarrow
i} = \alpha \tilde{r}_i + \frac{(\max_k (\tilde{r}_k) -
\tilde{r}_i)U_i}{ \sum_{k=1}^{N}\tilde{r}_k - \tilde{r}_i}.\]The
rest of the source node's uplink is
\[U_s - \sum_{i=1}^{N}r^{(2)}_{s \rightarrow i} = U_s - \alpha \max (\tilde{r}_k).\]The optimal depth-1 tree
can be obtained by the convex optimization
\begin{equation} \label{eq:type1treeminproblem}
  \begin{array}{cc}
    \min & \sum_{i=1}^{N} W_i \frac{B}{r_i} \\
    \textrm{subject to} & r_i = \beta_i + r^{(1)}_{s \rightarrow i},\\
     & r^{(1)}_{s \rightarrow i} \geq 0, \forall i=1\cdots, N,\\
     & \sum_{i=1}^{N} r^{(1)}_{s \rightarrow i} \leq U_s - \alpha \max (\tilde{r}_k).\\
  \end{array}
\end{equation}
The optimal solution to the problem (\ref{eq:type1treeminproblem})
is
\begin{equation}
r_i = \Bigg \{ \begin{array}{cc}
                 \sqrt{W_i}\cdot R, & \textrm{if } \sqrt{W_i}\cdot R \geq \beta_i, \\
                 \beta_i & \textrm{if } \sqrt{W_i}\cdot R <
                 \beta_i,
               \end{array}
              \label{eq:waterfilling3}
\end{equation}
and
\begin{equation}
r^{(1)}_{s \rightarrow i} = \Bigg \{ \begin{array}{cc}
                 \sqrt{W_i}\cdot R - \beta_i, & \textrm{if } \sqrt{W_i}\cdot R \geq \beta_i, \\
                 0 & \textrm{if } \sqrt{W_i}\cdot R < \beta_i,
               \end{array}
              \label{eq:waterfilling4}
\end{equation}
where $R$ is chosen such that $\sum_{i=1}^{N} r^{(1)}_{s \rightarrow
i} = U_s - \alpha \max (\tilde{r}_k)$ (also $\sum_{i=1}^{N}r_i =U_s
+ \sum_{i=1}^{N}U_i$).

The complexity of calculating this suboptimal network resource
allocation is $O(N^2)$. Note that when $W_i = 1$ for all
$i=1,\cdots, N$, this suboptimal network resource allocation is the
same as that of Mutualcast, and hence, this network resource
allocation is optimal for this case. For general weight settings,
this network resource allocation guarantees that the flow rate to
peer $i$ is larger than or equal to $\tilde{r}_i$, which is stated
in the following theorem.
\newtheorem{almostoptimal}[StaticPolynomial]{Theorem}
\begin{almostoptimal}\label{Theorem:almostoptimal}
For P2P networks with peer uplink constraints but no peer downlink
constraints (i.e. $D_i= \infty$), the network resource allocation
determined by (\ref{eq:rateassign1}) (\ref{eq:rateassign2})
(\ref{eq:c}) (\ref{eq:waterfilling3}) and (\ref{eq:waterfilling4})
ensures that the WSDT $\sum_{i=1}^{N}W_iB/r_i$ is less than or equal
to the WSDT associated to (\ref{eq:waterfilling2}), i.e.,
$\sum_{i=1}^{N}W_iB/\tilde{r}_i$.
\end{almostoptimal}
\begin{proof}
If $c=\frac{U_s}{ \alpha \max(\tilde{r}_k)}$, the flow rate to peer
$i$ is
\begin{align}
r_i & = r^{(2)}_{s \rightarrow i} + \sum_{j=1,j\neq i}^{N}r_{j
\rightarrow i}\\
& = c \alpha \tilde{r}_i + c \frac{(\max_k (\tilde{r}_k) -
\tilde{r}_i)U_i}{ \sum_{k=1}^{N}\tilde{r}_k - \tilde{r}_i} \\
& \geq c \alpha \tilde{r}_i\\
& = \frac{U_s}{\max(\tilde{r}_k)} \tilde{r}_i \\
& \geq \tilde{r}_i, \label{eq:5}
\end{align}
where (\ref{eq:5}) follows from $\tilde{r}_k \leq \tilde{D}_k \leq
U_s$. If $c=1$, a feasible solution to problem
(\ref{eq:type1treeminproblem}) is
\[r^{(1)}_{s \rightarrow i} = (U_s - \alpha \max(\tilde{r}_k)) \frac{\tilde{r}_i}{\sum_{k=1}^{N}\tilde{r}_k}.\]
For this feasible solution, the total flow rate to peer $i$ with the
depth-1 tree and the depth-2 trees is
\begin{align}
r_i & = \beta_i + (U_s - \alpha \max(\tilde{r}_k))
\frac{\tilde{r}_i}{\sum_{k=1}^{N}\tilde{r}_k} \\
& = (\alpha + (U_s - \alpha \max(\tilde{r}_k))
\frac{1}{\sum_{k=1}^{N}\tilde{r}_k}) \tilde{r}_i +
\frac{(\max(\tilde{r}_k) -
\tilde{r}_i)U_i}{\sum_{k=1}^{N}\tilde{r}_k - \tilde{r}_i}.
\end{align}
Denote $\gamma = \alpha + (U_s - \alpha \max(\tilde{r}_k))
\frac{1}{\sum_{k=1}^{N}\tilde{r}_k}$. We have
\begin{align}
U_s + \sum_{i=1}^{N}U_i & = \sum_{i=1}^{N} r_i \label{eq:6}\\
& = \gamma \sum_{i=1}^{N}\tilde{r}_i + \sum_{i=1}^{N}
\frac{(\max(\tilde{r}_k) -
\tilde{r}_i)U_i}{\sum_{k=1}^{N}\tilde{r}_k - \tilde{r}_i} \\
& \leq \gamma \sum_{i=1}^{N}\tilde{r}_i + \sum_{i=1}^{N}
\frac{\max(\tilde{r}_k)
U_i}{\sum_{k=1}^{N}\tilde{r}_k} \label{eq:7}\\
& \leq \gamma \sum_{i=1}^{N}\tilde{r}_i +
\max(\tilde{r}_k) \label{eq:8}\\
& = \gamma (U_s + \sum_{i=1}^{N}U_i - \max(\tilde{r}_k)) +
\max(\tilde{r}_k).\label{eq:9}
\end{align}
Some of these steps are justified as follows:
\begin{itemize}
\item (\ref{eq:6}) follows from the fact that all uplink resource is
deployed;
\item (\ref{eq:7}) follows from the inequality
$\frac{b-d}{a-d} \leq \frac{b}{a}$ when $a \geq b \geq d \geq 0$;
\item(\ref{eq:8}) follows from $\sum_{k=1}^{N}\tilde{r}_k = U_s - \max{\tilde{r}_k} + \sum_{i=1}^{N}U_i \geq \sum_{i=1}^{N}U_i \geq U_i$.
\end{itemize}
Therefore, $\gamma \geq 1$. Hence,
\begin{equation}
r_i  = \gamma \tilde{r}_i + \frac{(\max(\tilde{r}_k) -
\tilde{r}_i)U_i}{\sum_{k=1}^{N}\tilde{r}_k - \tilde{r}_i} \geq
\tilde{r}_i,
\end{equation}
which indicates that this feasible solution $r_i$ to the problem
(\ref{eq:type1treeminproblem}) provides a WSDT less than or equal to
$\sum_{i=1}^{N}W_iB/\tilde{r}_i$. Hence, the network resource
allocation determined by (\ref{eq:rateassign1})
(\ref{eq:rateassign2}) (\ref{eq:c}) (\ref{eq:waterfilling3}) and
(\ref{eq:waterfilling4}) also provides a WSDT less than or equal to
$\sum_{i=1}^{N}W_iB/\tilde{r}_i$.
\end{proof}

\subsection{Resource Allocation with Peer Downlink Constraints}

Now we consider P2P networks with both peer uplink bandwidth
constraints and peer downlink bandwidth constraints. The idea of the
resource allocation for these P2P networks is the same as that for
P2P networks without downlink constraints. The details are provided
as follows:

If $U_s + \sum_{i=1}^{N}U_i \geq \sum_{i=1}^{N}\tilde{D}_i$, from
Corollary \ref{Corollary:EnoughUplink}, Algorithm
\ref{alg:extendedMutualcast} provides the optimal network resource
allocation.

If $U_s + \sum_{i=1}^{N}U_i < \sum_{i=1}^{N}\tilde{D}_i$,
consider a water-filling-type solution
\begin{equation}
\tilde{r}_i = \Bigg \{ \begin{array}{cc}
                 \sqrt{W_i}\cdot R, & \textrm{if } \sqrt{W_i}\cdot R < \tilde{D}_i, \\
                 \tilde{D}_i & \textrm{if } \sqrt{W_i}\cdot R \geq \tilde{D}_i,
               \end{array}
              \label{eq:waterfilling5}
\end{equation}
where $R$ is chosen such that $\sum_{i=1}^{N}\tilde{r}_i = U_s +
\sum_{i=1}^{N}U_i - \max(\tilde{r}_k)$.

First construct the depth-2 trees with rates in
(\ref{eq:rateassign1}) and (\ref{eq:rateassign2}), where $c$ is
still chosen to be the largest possible value. However, for general
P2P networks, the constraints on $c$ are not only
(\ref{eq:c_constraint1}) (\ref{eq:c_cosntraint2}), but also
\begin{equation}
\beta_i \triangleq r^{(2)}_{s \rightarrow i} + \sum_{j=1,j\neq i}^{N}r_{j \rightarrow i} \leq \tilde{D}_i. \label{eq:c_constraint3}\\
\end{equation}
After constructing the depth-2 trees, the flow rate to peer $i$ is
$\beta_i$. The used source node's uplink is $ c \alpha \max
(\tilde{r}_k)$. If $ c \alpha \max (\tilde{r}_k) < U_s$, we can
further use the rest of the source node's uplink to distribute
content through the depth-1 tree. The optimal resource allocation
for the depth-1 tree can be obtained by the convex optimization
\begin{equation} \label{eq:type1treeminproblem2}
  \begin{array}{cc}
    \min & \sum_{i=1}^{N} W_i \frac{B}{r_i} \\
    \textrm{subject to} & r_i = \beta_i + r^{(1)}_{s \rightarrow i},\\
     & r^{(1)}_{s \rightarrow i} \geq 0, \forall i=1\cdots, N,\\
     & r_i \leq \tilde{D}_i, \forall i=1\cdots, N,\\
     & \sum_{i=1}^{N} r^{(1)}_{s \rightarrow i} \leq U_s - c \alpha \max (\tilde{r}_k).\\
  \end{array}
\end{equation}
The optimal solution to the problem (\ref{eq:type1treeminproblem2})
is
\begin{equation}
r_i = \Bigg \{ \begin{array}{cc}
                 \sqrt{W_i}\cdot R, & \textrm{if } \beta_i \leq \sqrt{W_i}\cdot R \leq \tilde{D}_i, \\
                 \beta_i & \textrm{if } \sqrt{W_i}\cdot R <
                 \beta_i,\\
                 \tilde{D}_i, & \textrm{if } \sqrt{W_i}\cdot R >
                 \tilde{D}_i,
               \end{array}
              \label{eq:waterfilling6}
\end{equation}
and
\begin{equation}
\tilde{r}^{(1)}_{s \rightarrow i} = r_i - \beta_i,
\label{eq:waterfilling7}
\end{equation}
where $R$ is chosen such that
$$\sum_{i=1}^{N} r_i = U_s + c\sum_{i=1}^{N}U_i.$$
The complexity of calculating this resource allocation is $O(N^2)$.

\subsection{Routing-Based Depth-2 Scheme}
\label{sec:Rountingscheme}

So far, this section has provided a family of rateless-coding-based
schemes for P2P file-transfer applications. In this subsection, we
introduce a routing-based scheme. This routing-based scheme is a
further extension to Extended Mutualcast. This scheme also applies
the tree structures in Fig.~\ref{fig:typetree} to distribute
content. The constraints on the network resource allocation for this
scheme are
\begin{equation}
r^{(1)}_{s \rightarrow k_1} \geq \cdots \geq r^{(1)}_{s \rightarrow
k_N} \geq 0, \label{eq:r_constraint1}
\end{equation}
and
\begin{align}
& r^{(2)}_{s \rightarrow k_i} \geq r_{k_i \rightarrow k_1} \geq
\cdots
\geq r_{k_i \rightarrow k_{i-1}} \nonumber \\
& \geq r_{k_i \rightarrow k_{i+1}} \geq \cdots \geq r_{k_i
\rightarrow k_N}, \forall i=1,\cdots,N, \label{eq:r_constraint2}
\end{align}
where $(k_1,\cdots, k_N)$ is the order in which the peers finish
downloading. In the rest of this subsection, we assume the order is
$(1,\cdots,N)$ for simplicity. These constraints are stricter than
those of the rateless-coding-based scheme, and they are introduced
to simplify the routing scheme. In particular, given the order of
$(1,\cdots,N)$ in which peers finish downloading, the proposed
routing-based scheme ensures that at any time in the scheme, peer
$i$ has all packets received by peers $i+1,\cdots, N$ for all
$i=1,\cdots,N-1$. This condition can be achievable if the network
resource allocation satisfies (\ref{eq:r_constraint1}) and
(\ref{eq:r_constraint2}). For the routing-based scheme, when peer
$i$ finishes downloading, the scheme starts to only broadcast the
chunks which peer $i+1$ hasn't received, called interesting chunks.
With this condition, the interesting chunks are also new to peers
$i+2, \cdots,N$. The details of the routing-based scheme is given in
Algorithm \ref{alg:routing}.

\begin{algorithm}
\caption{Routing-Based Scheme} \label{alg:routing}
\begin{algorithmic}[1]
\STATE Given the order in which the peers finish downloading. Assume
the order is $(1,\cdots,N)$ for simplicity. \\
\STATE Given the network resource allocation $\{r_{i \rightarrow
j}$, $r^{(1)}_{s \rightarrow i}, r^{(2)}_{s \rightarrow
i}\}$ for $i,j =1,\cdots,N$, which satisfies the constraints  (\ref{eq:r_constraint1}) and (\ref{eq:r_constraint2}).\\
\STATE Partition the whole file into many chunks. \\
\FOR {Step $i=1$ to $N$} \STATE At the beginning of Step $i$, peer
$1,\cdots, i$
finish downloading. \\
\STATE In Step $i$, only broadcast the chunks which peer $i$ doesn't
have, called interesting chunks.
Note that all peers $i, \cdots , N$ don't contain the interesting chunks. \\
\STATE Distribute interesting chunks along the depth-1 tree and
the depth-2 trees according to the network resource allocation.\\
\STATE For the depth-1 tree , the set of chunks sent to peer $i$
contains the set of chunks sent to peer
$j$ for $i<j$. \\
\STATE For the depth-2 tree in Fig.~\ref{fig:typetree}(b), the set
of chunks from peer $i$ to peer $k$ contains the set of chunks from
peer $i$ to peer $j$ for $k>j$. Peer $i$ only keeps the set of
chunks sent to peer $i-1$ for $i=2,\cdots,N$.\\
\STATE The above two chunk selection constraints guarantee that
peers $i, \cdots , N$ don't contain the interesting chunks in Step
$i$ for $i=1,\cdots, N$. \STATE Until peer $i$ receives all
interesting chunks and finishes downloading.\\ \ENDFOR
\end{algorithmic}
\end{algorithm}
The optimal network resource allocation for this routing-based
scheme can be obtained by the convex optimization of minimizing
$\sum_{i=1}^{N}W_iB/r_i$ subject to the constraints
(\ref{eq:r_constraint1}) (\ref{eq:r_constraint2}), nodes' uplink and
downlink constraints, and the flow rate expression
$$r_i = \sum_{j=1,j\neq i}^{N}r_{j \rightarrow
i} + r^{(1)}_{s \rightarrow i} + r_{i \rightarrow i-1}, \quad
i=1,\cdots,N,$$ where $r_{1 \rightarrow 0} = r^{(2)}_{s \rightarrow
1}$. The complexity for the interior point method to solve the
problem is $O((N^2)^{3.5})$. For the case of $W_i=1$ and $D_i =
\infty$, the optimal network resource allocation is the same as that
of Mutualcast. For the case of $W_i=1$ or $U_s + \sum_{i=1}^{N}U_i
\geq \sum_{i=1}^{N}\tilde{D}_i$, by Theorem \ref{Theorem:SumDelay}
and Corollary \ref{Corollary:EnoughUplink}, Algorithm
\ref{alg:extendedMutualcast} provides the optimal network resource
allocation.

For general cases with $U_s + \sum_{i=1}^{N}U_i <
\sum_{i=1}^{N}\tilde{D}_i$, we provide a suboptimal network resource
allocation for this routing-based scheme. Consider the
water-filling-type solution in (\ref{eq:waterfilling5}). Without
loss of generality, assume that $\tilde{r}_1 \geq \cdots \geq
\tilde{r}_N$, and give the ordering $(1,\cdots,N)$ in which the
peers finish downloading. First construct the depth-2 trees with
rates in (\ref{eq:rateassign1}) and (\ref{eq:rateassign2}), where
$c$ is still chosen to be the largest possible value satisfying
(\ref{eq:c_constraint1}) (\ref{eq:c_cosntraint2}) and
(\ref{eq:c_constraint3}). After constructing the depth-2 trees, the
effective flow rate to peer $i$ is
\begin{align}
\beta_i & = \sum_{j=1,j\neq i}^{N}r_{j \rightarrow i} + r_{i
\rightarrow i-1} \\
& = c (\alpha \tilde{r}_i + \frac{\tilde{r}_{i-1} -
\tilde{r}_i}{\sum_{k=1}^{N}\tilde{r}_k - \tilde{r}_i}U_i),
\label{eq:betai}
\end{align}
where $\tilde{r}_0 \triangleq \tilde{r}_1$. The download rate (used
downlink) for peer $i$ is $c (\alpha \tilde{r}_i +
\frac{\tilde{r}_{1} - \tilde{r}_i}{\sum_{k=1}^{N}\tilde{r}_k -
\tilde{r}_i}U_i)$. Note that the effective flow rate is smaller than
the download rate for peer $i$. This is because peer $i$ only keeps
a subset of chunks received from the source node. For this reason,
parts of peer $i$'s downlink and the source node's uplink are
wasted. The total amount of the wasted uplink is
\begin{equation}
U_w = c \sum_{i=1}^{N} \frac{\tilde{r}_1 -
\tilde{r}_{i-1}}{\sum_{k=1}^{N}\tilde{r}_k - \tilde{r}_i}U_i.
\end{equation}
The used source node's uplink is $ c \alpha \tilde{r}_1$. If $ c
\alpha \tilde{r}_1 < U_s$, we can further use the rest of the source
node's uplink to distribute content through the depth-1 tree. The
constraints on the resource allocation for the depth-1 tree are
(\ref{eq:r_constraint1}),
\begin{equation}
r^{(1)}_{s \rightarrow i} \leq D_i - \beta_i, \forall i=1,\cdots,N,
\end{equation}
and
\begin{equation}
\sum_{i=1}^{N}r^{(1)}_{s \rightarrow i} \leq U_s - c \alpha
\tilde{r}_1.
\end{equation}
Let $\hat{W}_i = \min_{k \leq i} (W_k)$. Let $\hat{D}_i = \min_{k
\leq i} (\tilde{D}_k - \beta_k)$. A sub-optimal network resource
allocation for the depth-1 tree is
\begin{equation}
r^{(1)}_{s \rightarrow i} = \Bigg \{ \begin{array}{cc}
                 \sqrt{\hat{W}_i}\cdot R - \beta_i, & \textrm{if } \beta_i \leq \sqrt{\hat{W}_i}\cdot R \leq \hat{D}_i, \\
                 0 & \textrm{if } \sqrt{\hat{W}_i}\cdot R <
                 \beta_i,\\
                 \hat{D}_i - \beta_i, & \textrm{if } \sqrt{\hat{W}_i}\cdot R >
                 \hat{D}_i,
               \end{array}
              \label{eq:waterfilling8}
\end{equation}
and $r_i = r^{(1)}_{s \rightarrow i} + \beta_i$, where $R$ is chosen
such that $$\sum_{i=1}^{N} r^{(1)}_{s \rightarrow i} = U_s - c
\alpha \tilde{r}_1$$ and also $$\sum_{i=1}^{N} r_i = U_s +
c\sum_{i=1}^{N}U_i - U_w.$$ The complexity of calculating the
suboptimal resource allocation for the routing-based scheme is
$O(N^2)$.

\subsection{Simulations for the Static Scenario}
\label{sec:staticsimulation} This subsection provides the empirical
WSDT performances of the rateless-coding-based scheme, the
routing-based scheme, and compares them with the lower bound to the
WSDT. In all simulations, the file size $B$ is normalized to be 1.
This subsection shows simulations for 6 cases of network settings as
follows:
\begin{itemize}
\item Case I: $U_i = 1$, $D_i = \infty$ for $i=1,\cdots,N$;
\item Case II: $U_i = 1$, $D_i = 8$ for $i=1,\cdots,N$;
\item Case III: $U_i = i/N$, $D_i = \infty$ for $i=1,\cdots,N$;
\item Case IV: $U_i = i/N$, $D_i = 8i/N$ for $i=1,\cdots,N$;
\item Case V: $U_i = 1 + 9\delta(i>N/2)$, $D_i = \infty$ for $i=1,\cdots,N$;
\item Case VI: $U_i = 1 + 9\delta(i>N/2)$, $D_i = 8i/N$, $i=1,\cdots,N$;
\end{itemize}
where $\delta(\cdot)$ is the indicate function.

Consider small P2P networks with $N=10$ peers. The performances of
sum download time versus $U_s$ for these 6 cases are shown in
Fig.~\ref{fig:staticsimulation1}. The performances of WSDT versus
$U_s$ with weight $W_i = i/N$ ($i=1,\cdots,N$) are shown in
Fig.~\ref{fig:staticsimulation2}. The performances of WSDT versus
$U_s$ with weight $W_i = 1 + \delta(i>N/2)$ ($i=1,\cdots,N$) are
shown in Fig.~\ref{fig:staticsimulation3}. In all these simulations,
the weighted sum download times of the rateless-coding-based scheme
and the routing-based scheme achieve or almost achieve the lower
bound.

\begin{figure}
  \centering
  \includegraphics[width=0.7\textwidth]{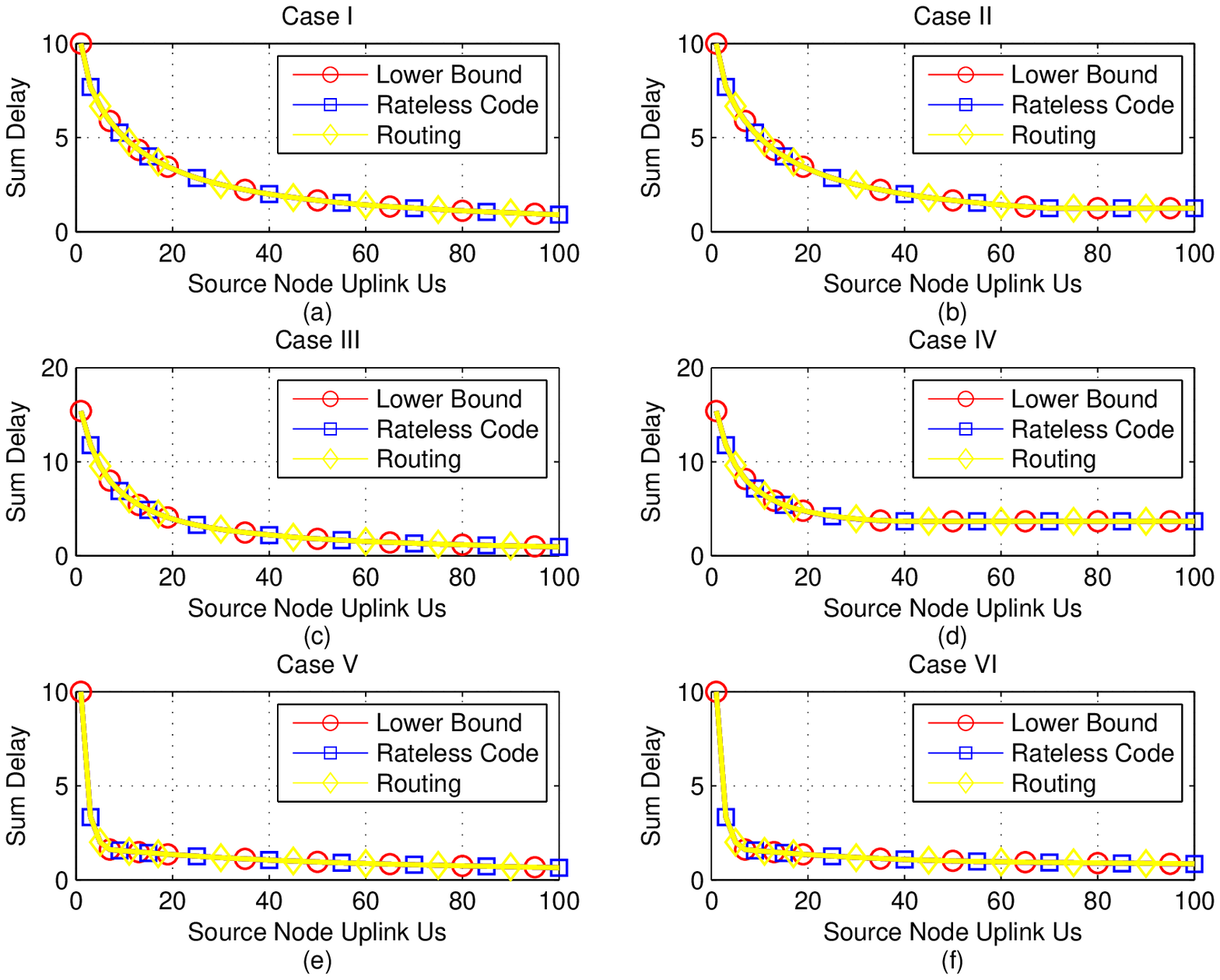}
  \caption{Sum download time versus $U_s$ for small P2P networks with $N=10$ peers.}\label{fig:staticsimulation1}
\end{figure}

\begin{figure}
  \centering
  \includegraphics[width=0.7\textwidth]{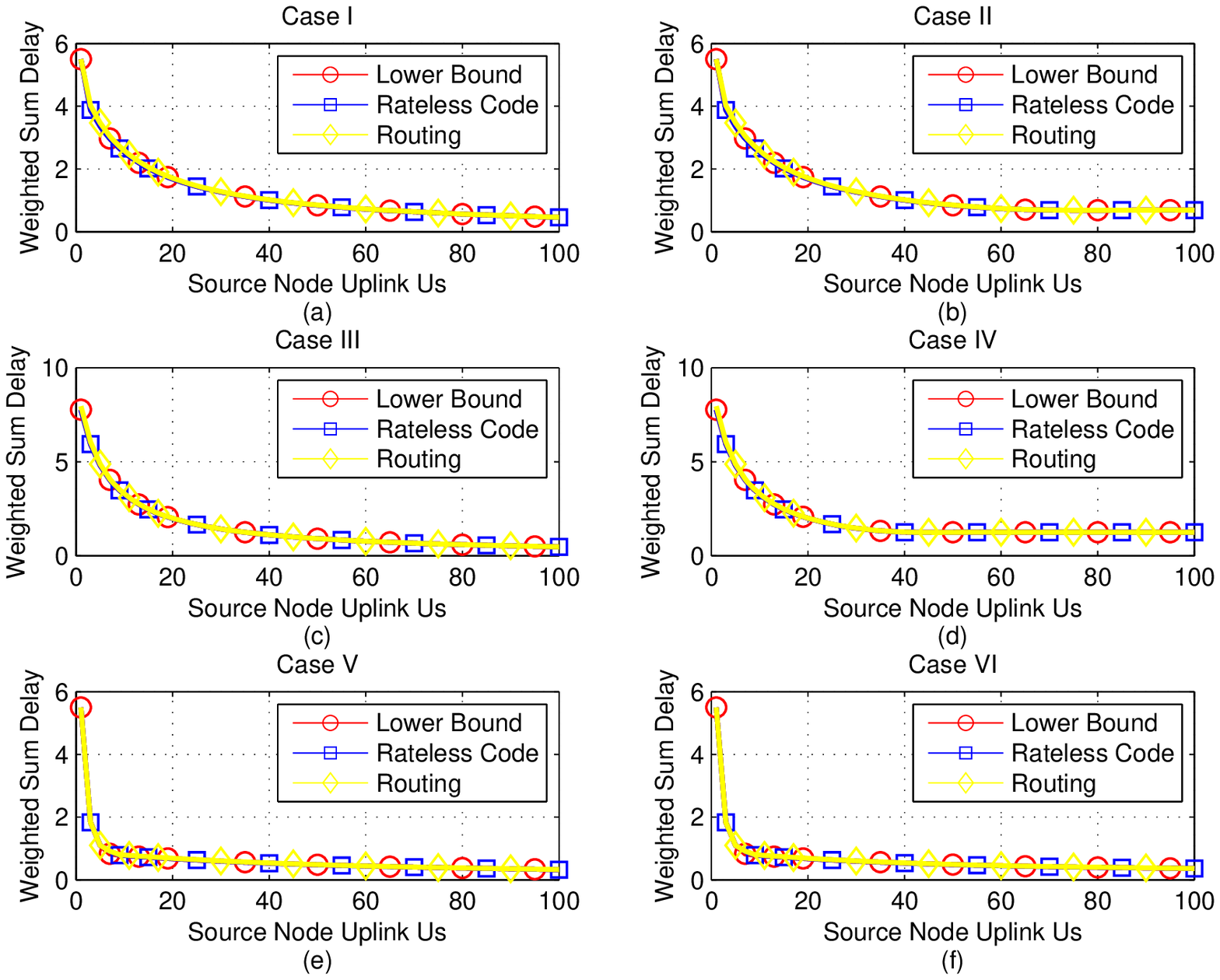}
  \caption{Weighted sum downloading time versus $U_s$ for small P2P networks with $N=10$ peers and weight $W_i = i/N$.}\label{fig:staticsimulation2}
\end{figure}

\begin{figure}
  \centering
  \includegraphics[width=0.7\textwidth]{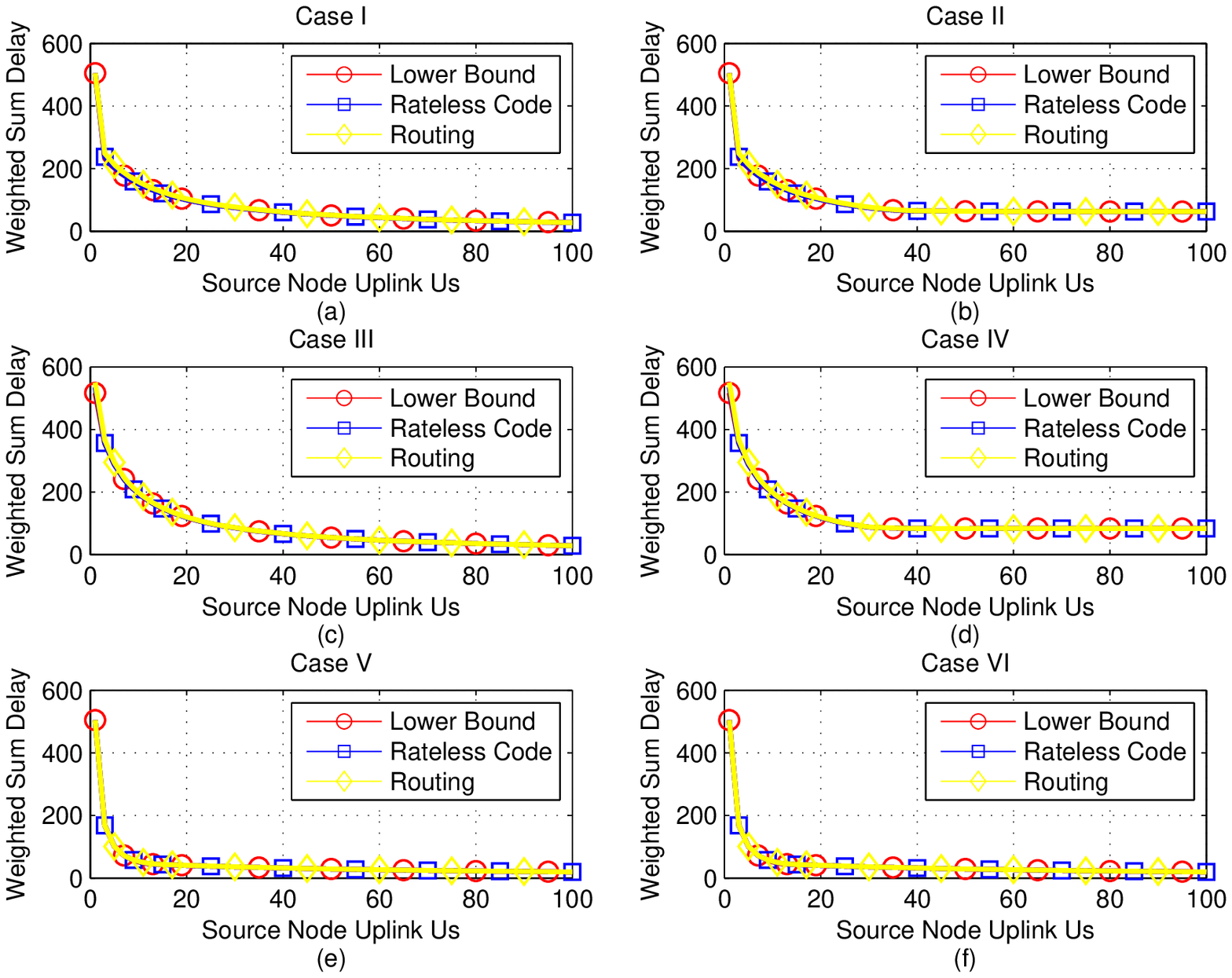}
  \caption{Weighted sum downloading time versus $U_s$ for small P2P networks with $N=10$ peers and weight $W_i = 1 + 99\delta(i>N/2)$.}\label{fig:staticsimulation3}
\end{figure}

Consider large P2P networks with $N=1000$ peers. The performances of
sum download time versus $U_s$ for these 6 cases are shown in
Fig.~\ref{fig:staticsimulation4}. The performances of WSDT versus
$U_s$ with weight $W_i = i/N$ ($i=1,\cdots,N$) are shown in
Fig.~\ref{fig:staticsimulation5}. The performances of WSDT versus
$U_s$ with weight $W_i = 1 + \delta(i>N/2)$ ($i=1,\cdots,N$) are
shown in Fig.~\ref{fig:staticsimulation6}. In all these simulations,
the weighted sum download times of the rateless-coding-based scheme
and the routing-based scheme also achieve or almost achieve the
lower bound.

\begin{figure}
  \centering
  \includegraphics[width=0.7\textwidth]{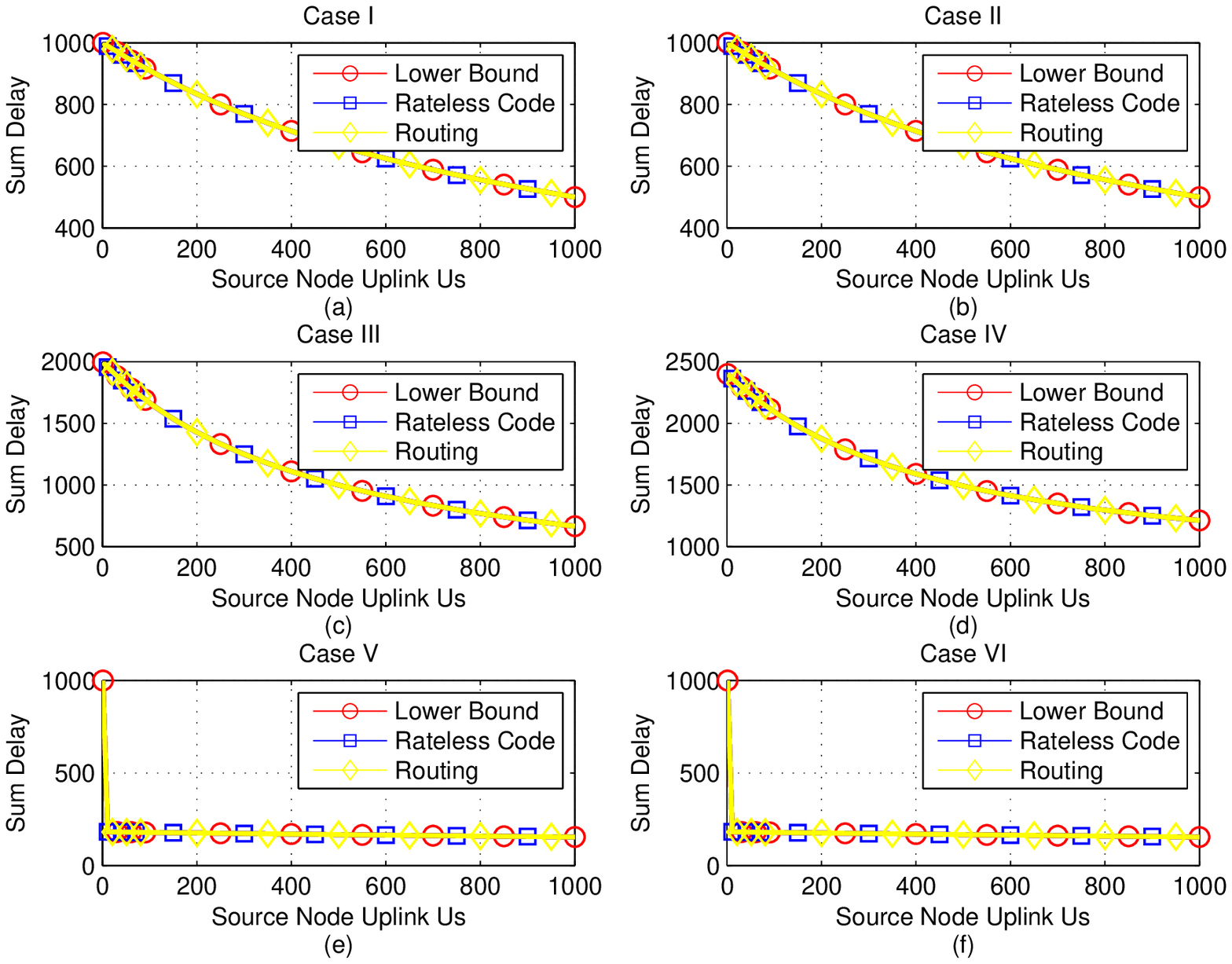}
  \caption{Sum download time versus $U_s$ for large P2P networks with $N=1000$ peers.}\label{fig:staticsimulation4}
\end{figure}

\begin{figure}
  \centering
  \includegraphics[width=0.7\textwidth]{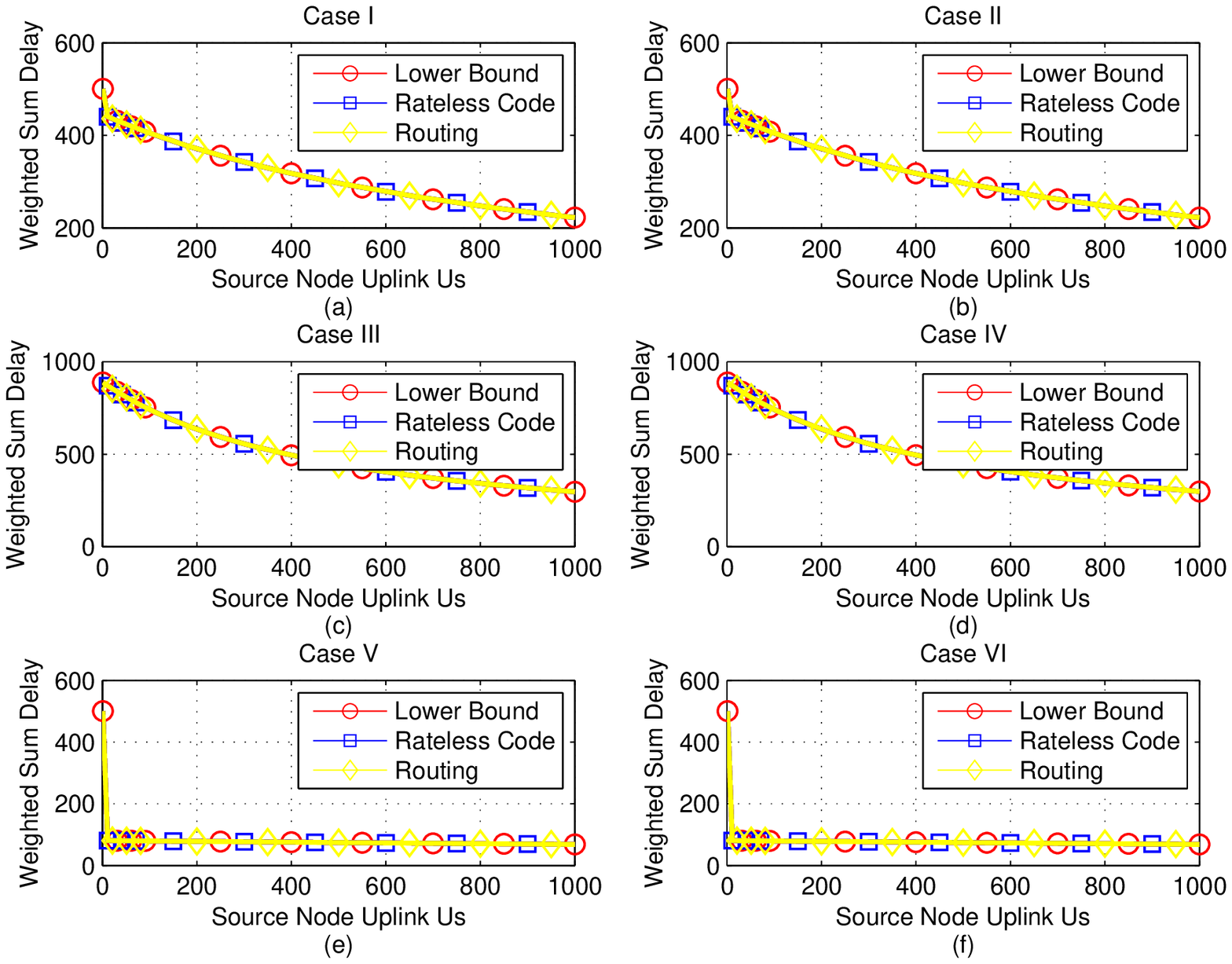}
  \caption{Weighted sum downloading time versus $U_s$ for large P2P networks with $N=1000$ peers and weight $W_i = i/N$.}\label{fig:staticsimulation5}
\end{figure}

\begin{figure}
  \centering
  \includegraphics[width=0.7\textwidth]{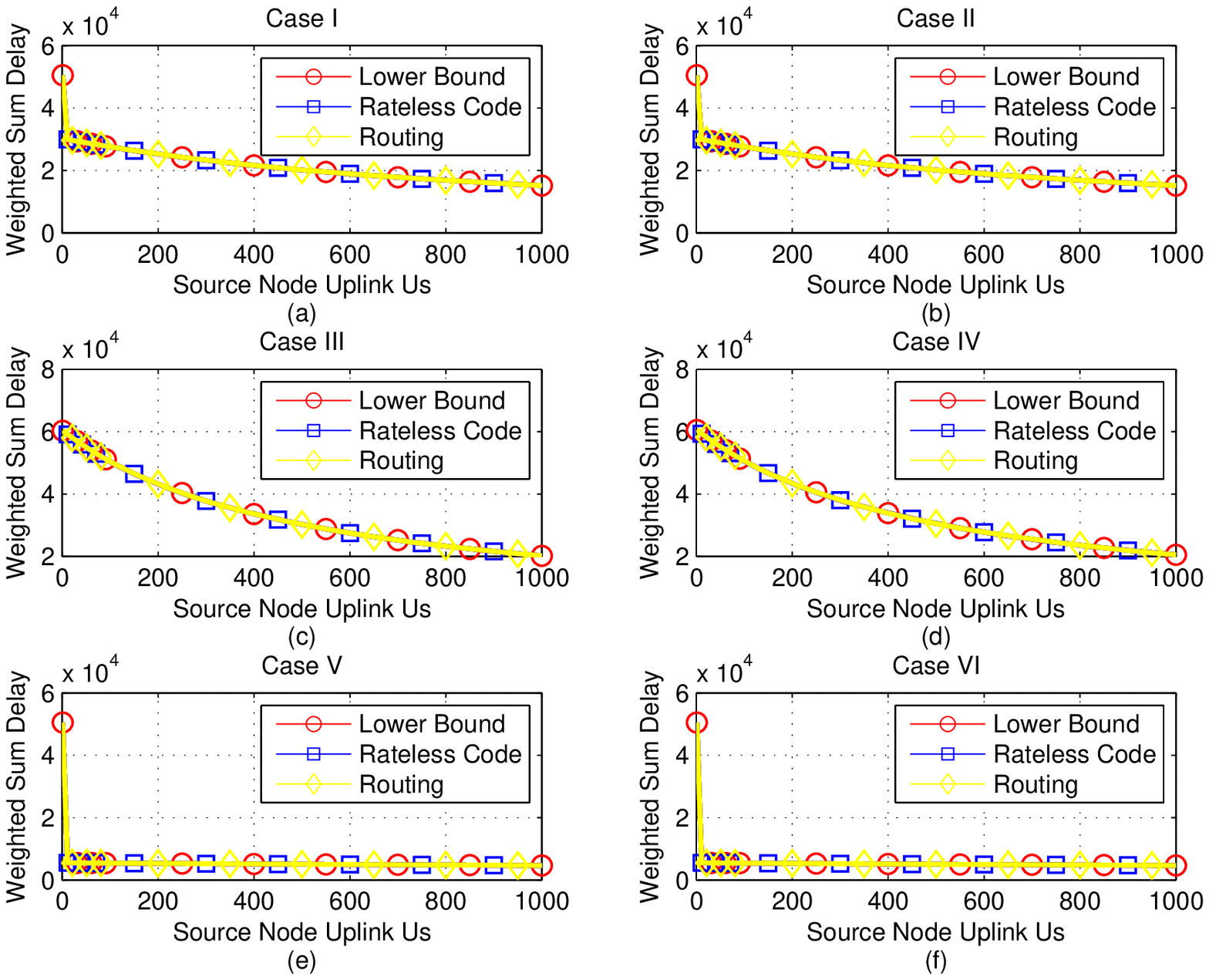}
  \caption{Weighted sum downloading time versus $U_s$ for large P2P networks with $N=1000$ peers and weight $W_i = 1 + 99\delta(i>N/2)$.}\label{fig:staticsimulation6}
\end{figure}

We also simulated for many other network settings and weight
settings. In all these simulations, the rateless-coding-based scheme
achieves or almost achieves the lower bound to the WSDT. Hence, the
lower bound to the WSDT is empirically tight, and the
rateless-coding-based scheme has almost-optimal empirical
performance. The routing-based scheme also has near-optimal
empirical performance.  However, for few cases there are clear
differences between the performance of the routing-based scheme and
the lower bound.

\section{The Dynamic Scenario}
\label{sec:dynamic}

The dynamic scenario is allowed to re-allocate the network resource
during the file transfer, in particular, whenever a peer finishes
downloading, joins into the network, or leaves from the network.

\subsection{A Piece-wise Static Approach to the General Dynamic Case}
Wu et al. \cite{Wu09} show that to optimize WSDT the network
resource allocation should be dynamic, but may remain constant
during any ``epoch'', a period of time between when one receiver
finishes downloading and another finishes downloading.  Thus, one
optimal solution for the dynamic scenario is ``piecewise static''.

As an example of how a ``piecewise static'' dynamic allocation can
reduce the WSDT, consider the example for which we studied static
allocations in Section \ref{sec:time-expanded-graph}.  Recall that
the example was for a P2P network with $U_S=2$, $B=1$ and three
peers $\{1,2,3\}$ with $U_1=U_2=U_3=1$ and $D_1 = D_2 = D_3 =
\infty$.   Fig.~\ref{fig:DynamicTimeExpandedGraph} shows the
time-expanded graph corresponding to the optimal dynamic rate
allocation for this example. Because there are three peers, this
time-expanded graph describes a file transfer scenario with 3
epochs. The first epoch lasts 0.5 unit time. In the first epoch, the
source node sends half of the file to peer 1 and the other half to
peer 2. Peer 1 and peer 2 exchange their received content, and
hence, both peer 1 and peer 2 finish downloading at the same time.
Hence, the second epoch lasts 0 time units (since $t_2-t-1=0$).  The
third epoch lasts 0.25 unit time, in which the source node, peer 1
and peer 2 transmits to peer 3 simultaneously. Peer 1 sends a
quarter of the file. Peer 2 sends another quarter. The source node
sends the other two quarters.

\begin{figure}
  \centering
  \includegraphics[width=0.5\textwidth]{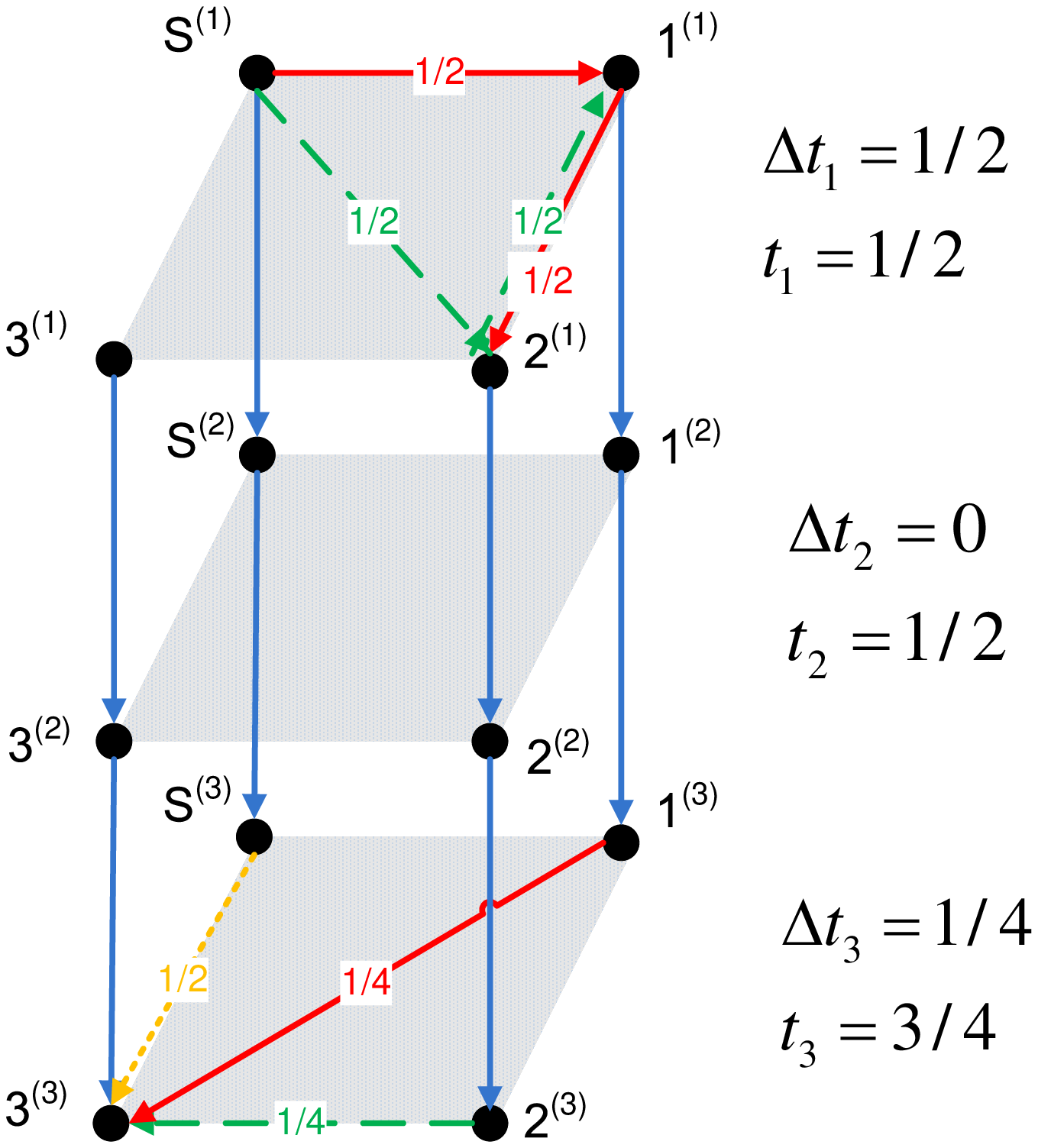}
  \caption{A time-expanded graph for a P2P network with three peers where $U_S = 2$, $B=1$, $U_1 = U_2 = U_3=1$ and $D_1 = D_2 = D_3 = \infty$.  Edges are labeled with the total information flow along the edge during the epoch.  This is the product of the rate allocation along the edge during the epoch and the duration of the epoch.}\label{fig:DynamicTimeExpandedGraph}
\end{figure}

This dynamic solution turns out to achieve the minimum possible sum
download time for this example which is 1.75.  For comparison, the
optimal static solution, which we saw in Section
\ref{sec:time-expanded-graph} had an only slightly larger sum
download time of 1.8.  This simple example shows that a dynamic rate
allocation can reduce WSDT.  In certain cases the benefit can be
significant.  Dynamic schemes can reduce the minimum sum download
time to approximately half that of the static case, at least when
downlink capacities are considered to be infinite \cite{Wu09}.

\subsection{A Rateless-coding Approach to Dynamic Allocation}

Wu et al. \cite{Wu09} propose a dynamic routing-based scheme. This
scheme first deploys all uplink resource to fully support the first
$K$ peers until they finish downloading, where $K$ is appropriately
chosen. After that, the scheme deploys all uplink resource to fully
support the next peer until it finishes downloading, an so forth.
Inspired by the work \cite{Wu09}, we propose a dynamic
rateless-coding-based scheme for P2P networks with both peer uplink
bandwidth constraints and peer downlink bandwidth constraints.  This
scheme is applicable for dynamic P2P networks in which peers may
even join or leave the network.

The key idea of this dynamic rateless-coding-based scheme is similar
to that of the dynamic routing-based scheme in \cite{Wu09}. In
particular, in each epoch, the scheme deploys all uplink resource to
fully support several chosen peers. The details of the dynamic
rateless-coding-based scheme are provided in Algorithm
\ref{alg:dynamic}.

\begin{algorithm}
\caption{Dynamic Rateless-Coding-Based Scheme} \label{alg:dynamic}
\begin{algorithmic}[1]
\STATE Initiate the P2P network. Peers join into the network.\\
\WHILE {A peer finishes downloading, joins into the network or
leaves from the network} \STATE Select a set of peers and reset
peers' weights. (The peer selection algorithm and the weight setting
are addressed in Section \ref{sec:ordering})\\
\STATE Apply the static rateless-coding-based scheme based on the
new weights until a peer finishes downloading, joins into the
network or leaves from the network. \ENDWHILE
\end{algorithmic}
\end{algorithm}

Algorithm \ref{alg:dynamic} provides the structure of the dynamic
rateless-coding-based scheme. Because the peers always receive
independently generated rateless coded chunks in the static
rateless-code scheme, the dynamic rateless-coding-based scheme is
also applicable for dynamic P2P network. As long as a peer receives
enough rateless coded chunks \footnote{The number of coded chunks
needed to decode the whole file is only slightly larger than the
total number of the original chunks.}, it can decode the whole file.
The key issue is how to set the peers' weights in each epoch. Since
the weight setting and the static rateless-coding-based scheme in
the current epoch will influence the dynamic scheme in the following
epoches, the problem of setting weights is very complicated. We will
address this problem in Section \ref{sec:ordering} and show that
this problem is approximately equivalent to selecting a set of peers
to fully support.

\subsection{A Solution to the Ordering Problem}
\label{sec:ordering}

Wu et al. \cite{Wu09} demonstrate that given an order in which the
receivers finish downloading, the dynamic allocation (neglecting
downlink bandwidth constraints) that minimizes WSDT can be obtained
in polynomial time by convex optimization and can be achieved
through linear network coding.   However, \cite{Wu09} leaves the
proper selection of the ordering as an open problem and does not
address the finite downlink capacities $D_i<\infty$ or the general
case of weighted sum download time which allows any values of the
weights $W_i$.

The simulations for the static scenario in Section
\ref{sec:staticsimulation} show that the WSDT of static
rateless-coding-based schemes are very close to that of the lower
bound (\ref{eq:lowerboundsolution}, \ref{eq:lowerbound}). Hence, the
flow rates $r_i$ in (\ref{eq:lowerboundsolution}) are achievable or
almost achievable by the static rateless-coding-based scheme. Recall
that the constraints on the rate $r_i$ in
(\ref{eq:lowerboundsolution}) are
$$0 \leq r_i \leq \tilde{D}_i, \quad \forall i=1,\cdots, N,$$
and
$$\sum_{i=1}^{N} r_i \leq U_s + \sum_{i=1}^{N}U_i.$$
In the following discussion, we assume that any set of flow rates
$r_i$ ($i=1,\cdots,N$) satisfying the above constraints is
achievable by the static rateless-coding-based scheme.

Consider one epoch of the dynamic rateless-coding-based scheme.
Suppose there are $N$ peers in the network in the current epoch.
Peer $i$ ($i=1,\cdots,N$) has uplink capacity $U_i$, downlink
capacity $D_i$ and $B - q_iB$ received rateless-coded chunks.
Suppose the static rateless-coding-based scheme supports peer $i$
with flow rate $r_i$ ($i=1,\cdots,N$) based on a weight setting. In
order to find the optimal weight setting for the current epoch, we
study the necessary conditions for the flow rates $r_i$
($i=1,\cdot,N$) to be optimal.

Let us first focus on two peers in the network, say peer 1 and peer
2. The total amount of the uplink resource supporting peer 1 and
peer 2 is $s = r_1 + r_2$. If the flow rates $r_i$ for
$i=1,\cdots,N$ is optimal, then the flow rates $r_1$ and $r_2$ are
also the optimal resource allocation for peers 1 and 2 given that
the flow rates $r_i$ for $i=3,\cdots,N$ are fixed. Now consider a
suboptimal scenario in which the uplink resource with the amount of
$s$ serves peers 1 and 2, and the rest of the uplink serves other
peers in all of the following epoches. This suboptimal scenario
provides a WSDT close to the minimum WSDT if $s \ll U_s +
\sum_{i=1}^{N}U_i$ (this is true for large $N$). Hence, we consider
this suboptimal scenario and address the necessary conditions for
$r_1$ and $r_2$ to be the optimal resource allocation for peers 1
and 2.

If $\frac{q_1B}{r_1} \leq \frac{q_2B}{r_2}$, then peer 1 finishes
downloading before peer 2 does. After peer 1 finishes downloading,
peer 1 acts as a source node and hence the total amount of the
source nodes' uplink is $U_s + U_1$, and peer 2 is supported by the
uplink resource with the amount of $s$. Hence, the WSDT for peers 1
and 2 is
\begin{equation}
\Delta_1 = W_1\frac{q_1B}{r_1} + W_2(\frac{q_1B}{r_1} + \frac{q_2B -
\frac{q_1B}{r_1}r_2}{\min (s, D_2,U_s +U_1)}), \label{eq:delta1}
\end{equation}
and
\begin{equation}
\frac{\Delta_1}{r_1} = \frac{q_1B}{r_1^2}(-W_1 - W_2 + \frac{s
W_2}{\min (s,D_2, U_s +U_1)}). \label{eq:deridelta1}
\end{equation}
Note that the sign of $\frac{\Delta_1}{r_1}$ does not depend on
$r_1$. Hence, the optimal solution to $r_1$ is either $r_1 =
r_2q_1/q_2$ (peer 1 and peer 2 finish at the same time) if $-W_1 -
W_2 + \frac{s W_2}{\min (s,D_2, U_s +U_1)} \geq 0$, or $r_1 = \min
(s,\tilde{D}_1)$ (peer 1 is fully supported) if $-W_1 - W_2 +
\frac{s W_2}{\min (s,D_2, U_s +U_1)} < 0$. Similarly, if
$\frac{q_1B}{r_1} \geq \frac{q_2B}{r_2}$, then peer 2 finishes
downloading before peer 1 does. The WSDT for peers 1 and 2 is
\begin{equation}
\Delta_2 = W_2\frac{q_2B}{r_2} + W_1(\frac{q_2B}{r_2} + \frac{q_1B -
\frac{q_2B}{r_2}r_1}{\min (s, D_1,U_s +U_2)}), \label{eq:delta2}
\end{equation}
and
\begin{equation}
\frac{\Delta_2}{r_2} = \frac{q_2B}{r_2^2}(-W_2 - W_1 + \frac{s
W_1}{\min (s,D_1, U_s +U_2)}). \label{eq:deridelta2}
\end{equation}
Note that the sign of $\frac{\Delta_2}{r_2}$ does not depend on
$r_2$ eithter. Hence, the optimal solution to $r_2$ is either $r_2 =
r_1q_2/q_1$ (peer 1 and peer 2 finish at the same time) if $-W_2 -
W_1 + \frac{s W_1}{\min (s,D_1, U_s +U_2)} \geq 0$, or $r_2 = \min
(s,\tilde{D}_2)$  (peer 2 is fully supported) if $-W_2 - W_1 +
\frac{s W_1}{\min (s,D_1, U_s +U_2)} < 0$. Therefore, the optimal
resource allocation for peer 1 and peer 2 is achieved when one of
the peers is fully supported, or they finish at the same time.
\newtheorem{fullsupport}[NetworkCoding]{Lemma}
\begin{fullsupport}\label{Lemma:fullsupport}
Given that the flow rates to peer $i$ for $i=3,\cdots,N$ are fixed,
and the amount of uplink resource supporting peer 1 and peer 2 is
$s$. If the optimal resource allocation for peer 1 and peer 2 is
achieved when they finish at the same time, then both peer 1 and
peer 2 are fully supported.
\end{fullsupport}
\begin{proof}
Let $\tilde{D}^{+}_1 = \min (s, D_1, U_s + U_2)$ and
$\tilde{D}^{+}_2 = \min (s, D_2, U_s + U_1)$. According to the above
discussion, the optimal resource allocation for peer 1 and peer 2 is
achieved when they finish at the same time if and only if $s \geq
\tilde{D}_1 + \tilde{D}_2$,
or $\frac{\Delta_1}{r_1} \geq 0$ and $\frac{\Delta_2}{r_2} \geq 0$.\\
If $s \geq \tilde{D}_1 + \tilde{D}_2$, then $r_1 = \tilde{D}_1$,
$r_2 = \tilde{D}_2$, and hence, peers 1 and 2 are
fully supported.\\
If $\frac{\Delta_1}{r_1} \geq 0$ and $\frac{\Delta_2}{r_2} \geq 0$,
then
$$-W_1 - W_2 + \frac{s W_2}{\tilde{D}^{+}_2} \geq 0,$$
and
$$-W_2 - W_1 + \frac{s W_1}{\tilde{D}^{+}_1} \geq 0$$.
Hence, $0 \leq W_1 \leq \frac{s -
\tilde{D}^{+}_2}{\tilde{D}^{+}_2}W_2$ and $0 \leq W_2 \leq \frac{s -
\tilde{D}^{+}_1}{\tilde{D}^{+}_1}W_1$. Multiply the above two
inequalities and obtain
$$s \geq \tilde{D}^{+}_1 + \tilde{D}^{+}_2 \geq \tilde{D}_1 + \tilde{D}_2.$$
Therefor, peer 1 and peer 2 are also fully supported.
\end{proof}

\newtheorem{fullsupport2}[EnoughUplink]{Corollary}
\begin{fullsupport2}\label{Corollary:fullsupport2}
Given that the flow rates to peer $i$ for $i=3,\cdots,N$ are fixed,
and the amount of uplink resource supporting peer 1 and peer 2 is
$s$. The optimal resource allocation for peer 1 and peer 2 is
achieved when one of them is fully supported or both of them are
fully supported.
\end{fullsupport2}

\newtheorem{fullsupport3}[EnoughUplink]{Corollary}
\begin{fullsupport3}\label{Corollary:fullsupport3}
The optimal network resource allocation in each epoch of a dynamic
scenario is only obtained when some peers are fully supported, at
most one peer is partially supported, and the other peers are not
supported.
\end{fullsupport3}
\begin{proof}
(\textbf{proof by contradiction}) If two peers are partially
supported, say peer 1 and peer 2 are partially supported, then the
resource allocation for peer 1 and peer 2 is not optimal by
Corollary \ref{Corollary:fullsupport2}.
\end{proof}

By Corollary \ref{Corollary:fullsupport3}, the optimal weight
setting in each epoch is $W=1$ for the fully supported peers, $0
\leq W \leq 1$ for the partially supported peer, and $W=0$ for other
peers. Hence, the problem of optimizing the weight setting is
approximately equivalent to selecting a set of peers to fully
support.

Now study the necessary conditions for a peer selection to be
optimal in a similar way. Suppose that the amount of uplink resource
supporting peer 1 and peer 2 is $s$, and
the flow rates to peer $i$ for $i=3,\cdots,N$ are fixed.\\
If $s < \frac{q_1+q_2}{q_1}\tilde{D}_1$ and $s < \frac{q_1 +
q_2}{q_2}\tilde{D}_2$, then peer 1 finishes downloading if peer 1 is
fully supported, or peer 2 finishes downloading if peer 2 is fully
supported. When peer 1 is fully supported, the WSDT for peer 1 and
peer 2 is $\Delta_1$ in (\ref{eq:delta1}) with $r_1 = \tilde{D}_1$.
When peer 2 is fully supported, the WSDT for these two peers is
$\Delta_2$ in (\ref{eq:delta2}) with $r_2 = \tilde{D}_2$. Hence, we
have
\begin{align}
\Delta_1 - \Delta_2 & = W_1\frac{q_1B}{\tilde{D}_1} +
W_2(\frac{q_1B}{\tilde{D}_1} +
\frac{q_2B - \frac{q_1B}{\tilde{D}_1}r_2}{\min (s, D_2,U_s +U_1)}) \nonumber \\
& - (W_2\frac{q_2B}{\tilde{D}_2} + W_1(\frac{q_2B}{\tilde{D}_2} +
\frac{q_1B - \frac{q_2B}{\tilde{D}_2}r_1}{\min (s, D_1,U_s +U_2)})) \\
& = q_1BW_1(\frac{1}{\tilde{D}_1} - \frac{1}{\tilde{D}^{+}_1}) - q_2BW_2 (\frac{1}{\tilde{D}_2} - \frac{1}{\tilde{D}^{+}_2}) \nonumber\\
& + q_1BW_2(\frac{1}{\tilde{D}^{+}_2} - \frac{1}{\tilde{D}_2})(1 - \frac{s}{\tilde{D}_1}) \nonumber \\
& - q_2BW_1(\frac{1}{\tilde{D}^{+}_1} - \frac{1}{\tilde{D}_1})(1 - \frac{s}{\tilde{D}_1}) \nonumber \\
& + (q_1BW_2 - q_2BW_1)(\frac{1}{\tilde{D}_1} + \frac{1}{\tilde{D}_2} - \frac{s}{\tilde{D}_1\tilde{D}_2}) \\
& \approx (\frac{W_2}{q_2} - \frac{W_1}{q_1}) \frac{q_1 q_2 B
(\tilde{D}_1 + \tilde{D}_2 -s)}{\tilde{D}_1 \tilde{D}_2}.
\end{align}
Therefore, it is better to first fully support peer 1  if
$\frac{W_2}{q_2} < \frac{W_1}{q_1}$ when $s <
\frac{q_1+q_2}{q_1}\tilde{D}_1$ and
$s < \frac{q_1 + q_2}{q_2}\tilde{D}_2$.\\
If $\frac{q_1}{\tilde{D}_1} > \frac{q_2}{\tilde{D}_2}$  and
$\frac{q_1 + q_2}{q_1}\tilde{D}_1 <s < \frac{q_1 +
q_2}{q_2}\tilde{D}_2$, then peer 2 always finishes downloading
before peer 1 does. In this case, it is better to first fully
support peer 1 if $\frac{\Delta_2}{r_2} >0$, i.e., $$\frac{W_1}{W_2}
> \frac{\tilde{D}^{+}_1}{s  - \tilde{D}^{+}_1},$$ or approximately
$$\frac{W_1}{W_2} > \frac{\tilde{D}_1}{s  - \tilde{D}_1}.$$\\
If $\frac{q_1}{\tilde{D}_1} < \frac{q_2}{\tilde{D}_2}$ and
$\frac{q_1 + q_2}{q_2}\tilde{D}_2 <s < \frac{q_1 +
q_2}{q_1}\tilde{D}_1$, then peer 1 always finishes downloading
before peer 2 does. In  this case, it is better to first fully
support peer 1 if $\frac{\Delta_1}{r_1} < 0$, i.e.,
$$\frac{W_1}{W_2} > \frac{\tilde{D}^{+}_2}{s  - \tilde{D}^{+}_2},$$
or approximately
$$\frac{W_1}{W_2} > \frac{\tilde{D}_2}{s  - \tilde{D}_2}.$$ \\
These discussions are concluded in the following theorem.
\newtheorem{orderfortwo}[StaticPolynomial]{Theorem}
\begin{orderfortwo}\label{Theorem:orderfortwo}
Given that the amount of uplink resource supporting peer $i$ and
peer $j$ is $s$, and the flow rates to peer $k$ for $k \neq i,j$ are
fixed. The optimal resource allocation for peer $i$ and peer $j$ is
to fully support peer $i$ (i.e., $r_i = \tilde{D}_i$) if
\begin{equation} \label{eq:ordercondition}
\frac{W_i}{W_j} > \Big \{ \begin{array}{cc}
                             \max(\frac{q_i}{q_j}, \frac{s  - \tilde{D}_j}{\tilde{D}_j}) & \textrm{when } \frac{q_i}{\tilde{D}_i} < \frac{q_j}{\tilde{D}_j}, \\
                             \frac{1}{\max(\frac{q_j}{q_i}, \frac{s  - \tilde{D}_i}{\tilde{D}_i})} & \textrm{when } \frac{q_i}{\tilde{D}_i} > \frac{q_j}{\tilde{D}_j}.
                           \end{array}
\end{equation}
\end{orderfortwo}

\newtheorem{orderforall}[EnoughUplink]{Corollary}
\begin{orderforall}\label{Corollary:orderforall}
Consider a peer selection for a dynamic scenario which selects peer
$i$ to fully support and peer $j$ to not support.  This peer
selection is optimal only if
\begin{equation} \label{eq:ordercondition2}
\frac{W_i}{W_j} > \Big \{ \begin{array}{cc}
                             \max(\frac{q_i}{q_j}, \frac{\tilde{D}_i  - \tilde{D}_j}{\tilde{D}_j}) & \textrm{when } \frac{q_i}{\tilde{D}_i} < \frac{q_j}{\tilde{D}_j}, \\
                             \frac{q_i}{q_j} & \textrm{when } \frac{q_i}{\tilde{D}_i} > \frac{q_j}{\tilde{D}_j}.
                           \end{array}
\end{equation}
\end{orderforall}
\begin{proof}
When peer $i$ is fully supported and peer $j$ is not supported,  $s
= r_i + r_j = \tilde{D}_i$. Plugging $s = \tilde{D}_i$ into
(\ref{eq:ordercondition}) and obtain (\ref{eq:ordercondition2}).
\end{proof}

Define the binary relation $\succ$ on $\{1, \cdots, N\}$ as $i \succ
j$ if (\ref{eq:ordercondition2}) is satisfied. Denote a peer
selection as $(I,J)$ where $I$ is the set of fully supported peers
and $J$ is the set of unsupported peers. $(I,J)$ is optimal only if
$i \succ j$ for any $i \in I$ and $j \in J$. For general P2P
networks, finding the optimal $(I,J)$ is computational impossible
because the binary relation $\succ$ is not transitive, which means
$$ i \succ j; j \succ k \nRightarrow i \succ k. $$

Define the binary relation $\succapprox$ on $\{1, \cdots, N\}$ as $i
\succapprox j$ if $\frac{W_i}{q_i} \geq \frac{W_j}{q_j}$. The binary
relation $\succapprox$ is an approximation to the binary relation
$\succ$. $i \succapprox j$ is equivalent to $i \succ j$ when
\begin{equation}
\frac{q_i}{q_j} > \frac{\tilde{D}_i - \tilde{D}_j}{\tilde{D}_j}.
\label{eq:binaryrelationequi}
\end{equation}
It can be seen by plugging (\ref{eq:binaryrelationequi}) into
(\ref{eq:ordercondition2}). The approximated binary relation
$\succapprox$ has the transitive property, and hence, the peers can
be ordered with respect to $\succapprox$. Based on this ordering, a
suboptimal peer selection algorithm and the corresponding weight
setting is constructed as shown in Algorithm \ref{alg:peerselect}.
\begin{algorithm}
\caption{Peer Selection and Weight Setting} \label{alg:peerselect}
\begin{algorithmic}[1]
\STATE Suppose $N$ peers are downloading in the current epoch.
\STATE Let $B - q_iB$ ($0<q_i\leq 1$)be the number of chunks that peer $i$ has received for $i=1,\cdots,N$.\\
\STATE Sort $\{\frac{W_i}{q_i}\}_{i=1}^{N}$ in descending order and get $(k_1,\cdots,k_N)$.\\
\STATE Find the smallest $M$ such that
$\sum_{i=1}^{M}\!\tilde{D}_{k_i}\! \geq \!U_s \!\!+ \!\!
\sum_{i=1}^{N}\!\!U_i$. \STATE Select peers $\{k_i\}_{i=1}^{M}$ to
fully support. \STATE Set $W_j = 1$ if $j \in \{k_i\}_{i=1}^{M}$, or
$W_j = 0$ otherwise.
\end{algorithmic}
\end{algorithm}

\section{Simulations of the Dynamic Scenario}
\label{sec:dynamicsimulation} The dynamic rateless-coding-based
scheme is feasible to both static P2P networks and dynamic P2P
networks. Consider a type of dynamic P2P networks which any peer
leaves from as it finishes downloading, and no peer joins into. This
section provides the empirical WSDT performances of the dynamic
rateless-coding-based scheme for static P2P networks and dynamic P2P
networks with peer leaving, and compares them with those of the the
static scenario for static P2P networks. In all simulations, the
file size $B$ is normalized to be 1. This section shows simulations
for Cases I,II,IV, and VI investigated in
\ref{sec:staticsimulation}.

\begin{figure}
  \centering
  \includegraphics[width=0.7\textwidth]{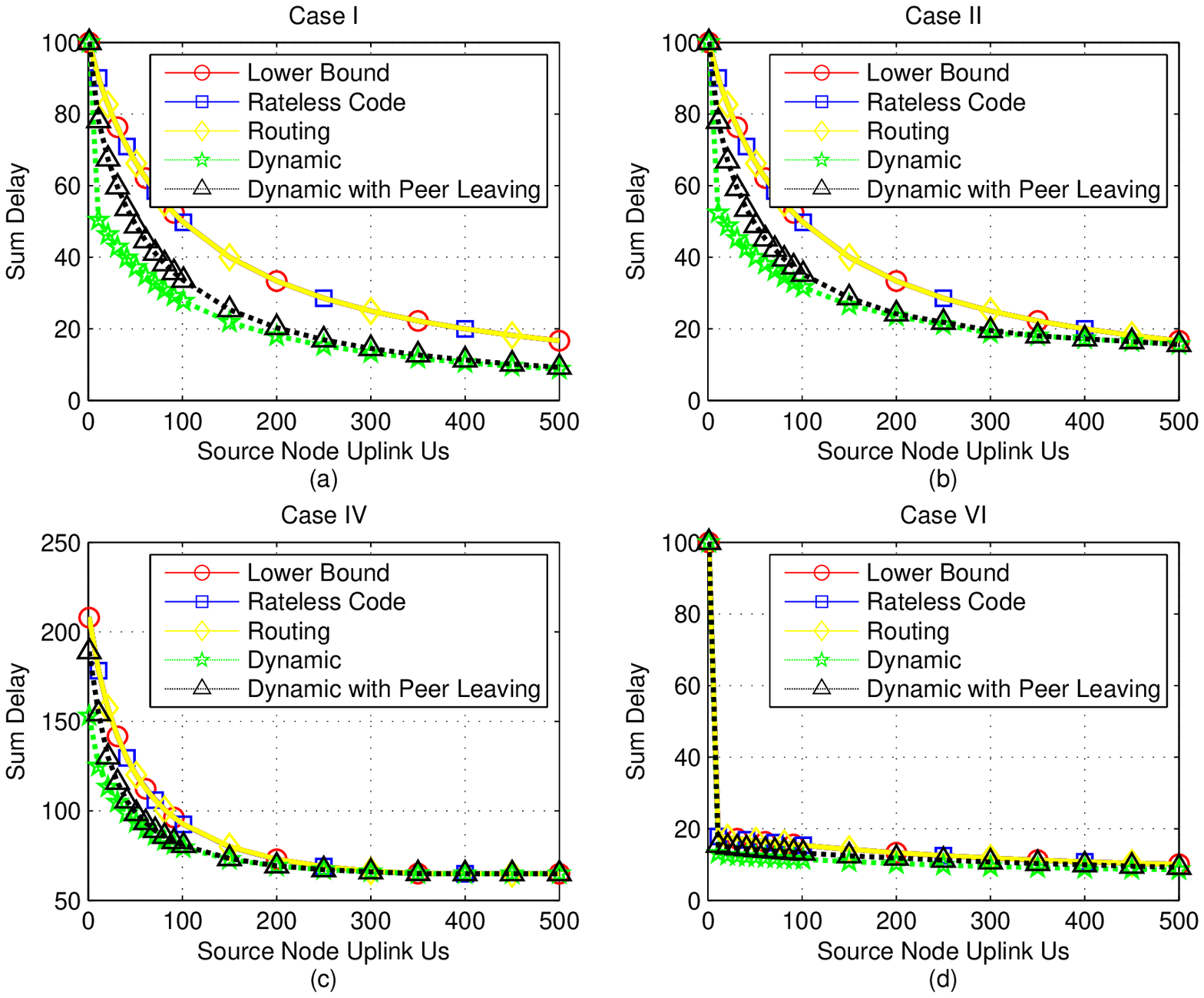}
  \caption{Sum download time versus $U_s$ for large P2P networks with $N=100$ peers.}\label{fig:dynamicsimulation1}
\end{figure}

\begin{figure}
  \centering
  \includegraphics[width=0.7\textwidth]{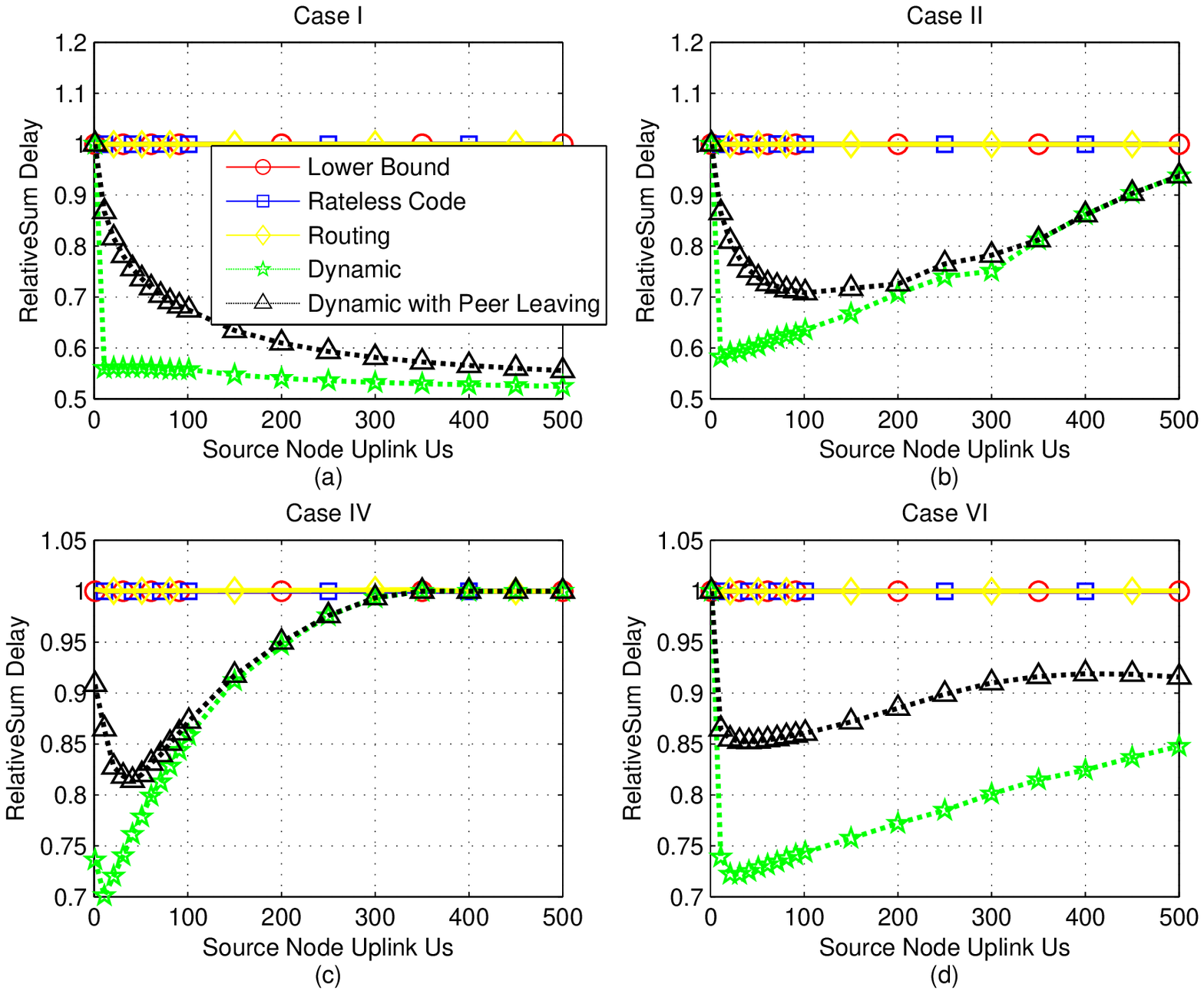}
  \caption{Relative sum download time versus $U_s$ for large P2P networks with $N=100$ peers.}\label{fig:dynamicsimulation11}
\end{figure}

Consider median-size P2P networks with $N=100$ peers. The
performances of sum download time versus $U_s$ for the 4 cases are
shown in Fig.~\ref{fig:dynamicsimulation1}.
Fig.~\ref{fig:dynamicsimulation11} shows the relative value of the
sum download time by normalizing the lower bound to be 1 in order to
explicitly compare the performances of the dynamic
rateless-coding-based scheme and the static scenario. For Case I
where peers have infinite downlink capacities, the sum download time
of the dynamic rateless-coding-based scheme is almost half of the
minimum sum download time for the static scenario for a broad range
of the source node uplink $U_s$. This result matches the results in
the previous work \cite{Wu09}, which says that the minimum sum
download time of dynamic scenarios is almost half of the minimum sum
download time of static scenarios when node uplinks are the only
bottleneck in the network. Our results also show that the sum
download time of the dynamic rateless-coding-based scheme with peer
leaving decreases to almost half of the minimum sum download time
for the static scenario as $U_s$ increases. For Cases II, IV, and
VI, the WSDs of the dynamic scheme and the dynamic scheme with peer
leaving are also always smaller than the minimum WSDT for the static
scenario. In particular, the WSDT of the dynamic scheme can be as
small as 0.59, 0.70, and 0.73 of the minimum WSDT for the static
scenario for Cases II, IV and VI, respectively. The WSDT of the
dynamic scheme with peer leaving can be as small as 0.71, 0.82, and
0.86 of the minimum WSDT for static scenarios for Cases II, IV and
VI, respectively. These largest improvements in percentage of
deploying the dynamic scheme is obtained when the source node can
directly support tens of the peers.

The performances of WSDT versus $U_s$ with weight $W_i = i/N$
($i=1,\cdots,N$) are shown in Fig.~\ref{fig:dynamicsimulation2}.
Fig.~\ref{fig:dynamicsimulation12} shows the relative value of the
WSDT. For Case I, the sum download times of the dynamic
rateless-coding-based scheme and the dynamic scheme with peer
leaving can be even less than half of the minimum sum download time
for the static scenario for a broad range of the source node uplink
$U_s$. This is because the peers with largest weight finish
downloading first in the dynamic scheme. The WSDT of the dynamic
scheme can be as small as 0.48, 0.49, and 0.58 of the minimum WSDT
for the static scenario for Cases II, IV and VI, respectively. The
WSDT of the dynamic scheme with peer leaving can be as small as
0.56, 0.62, and 0.77 of the minimum WSDT for the static scenario for
Cases II, IV and VI, respectively. Note that for Case VI, the WSDT
of the dynamic scheme with peer leaving is larger than that of the
static scenario for small $U_s$. This is because the peers with
larger uplink resource also have larger weight, and they finish
downloading and leave from the network first.

\begin{figure}
  \centering
  \includegraphics[width=0.7\textwidth]{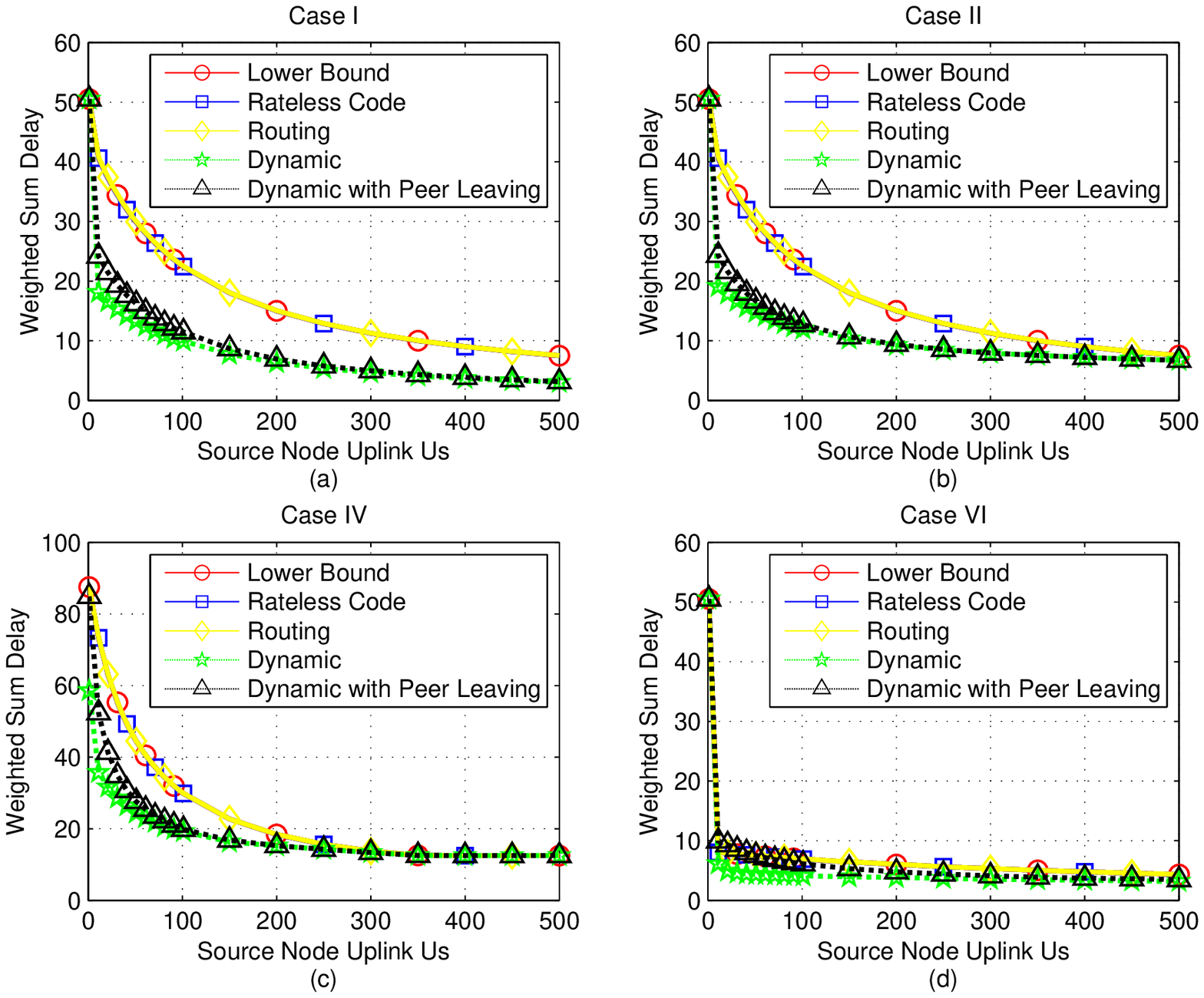}
  \caption{Weighted sum downloading time versus $U_s$ for large P2P networks with $N=100$ peers and weight $W_i = i/N$.}\label{fig:dynamicsimulation2}
\end{figure}

\begin{figure}
  \centering
  \includegraphics[width=0.7\textwidth]{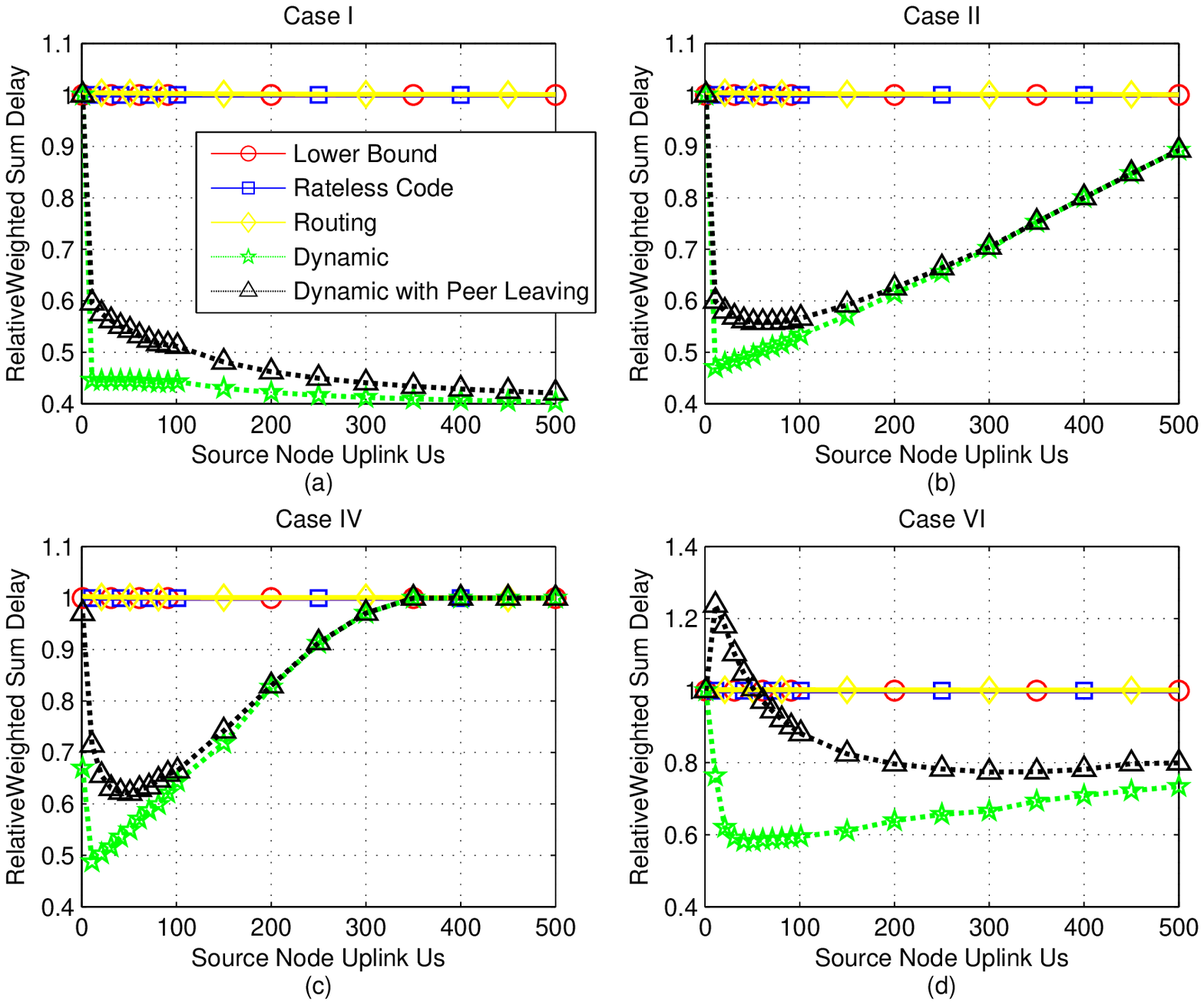}
  \caption{Relative weighted sum downloading time versus $U_s$ for large P2P networks with $N=100$ peers and weight $W_i = i/N$.}\label{fig:dynamicsimulation12}
\end{figure}

The performances of WSDT versus $U_s$ with weight $W_i = 1 +
99\delta(i  > N/2)$ ($i=1,\cdots,N$) are shown in
Fig.~\ref{fig:dynamicsimulation3}.
Fig.~\ref{fig:dynamicsimulation13} shows the relative value of the
WSDT. For Case I, the sum download times of the dynamic
rateless-coding-based scheme and the dynamic scheme with peer
leaving is around half of the minimum sum download time for the
static scenario for a broad range of the source node uplink $U_s$.
The WSDT of the dynamic scheme can be as small as 0.58, 0.55, and
0.52 of the minimum WSDT for static scenarios for Cases II, IV and
VI, respectively. The WSDT of the dynamic scheme with peer leaving
can be as small as 0.64, 0.64, and 0.63 of the minimum WSDT for the
static scenario for Cases II, IV and VI, respectively. Note that for
this weight setting, the WSDT of the dynamic scheme with peer
leaving is always smaller than that of the static scenario for Case
VI. This is because the gain by finishing peers with larger weight
is more than than the loss by the peers with larger uplink resource
leaving from the network.

\begin{figure}
  \centering
  \includegraphics[width=0.7\textwidth]{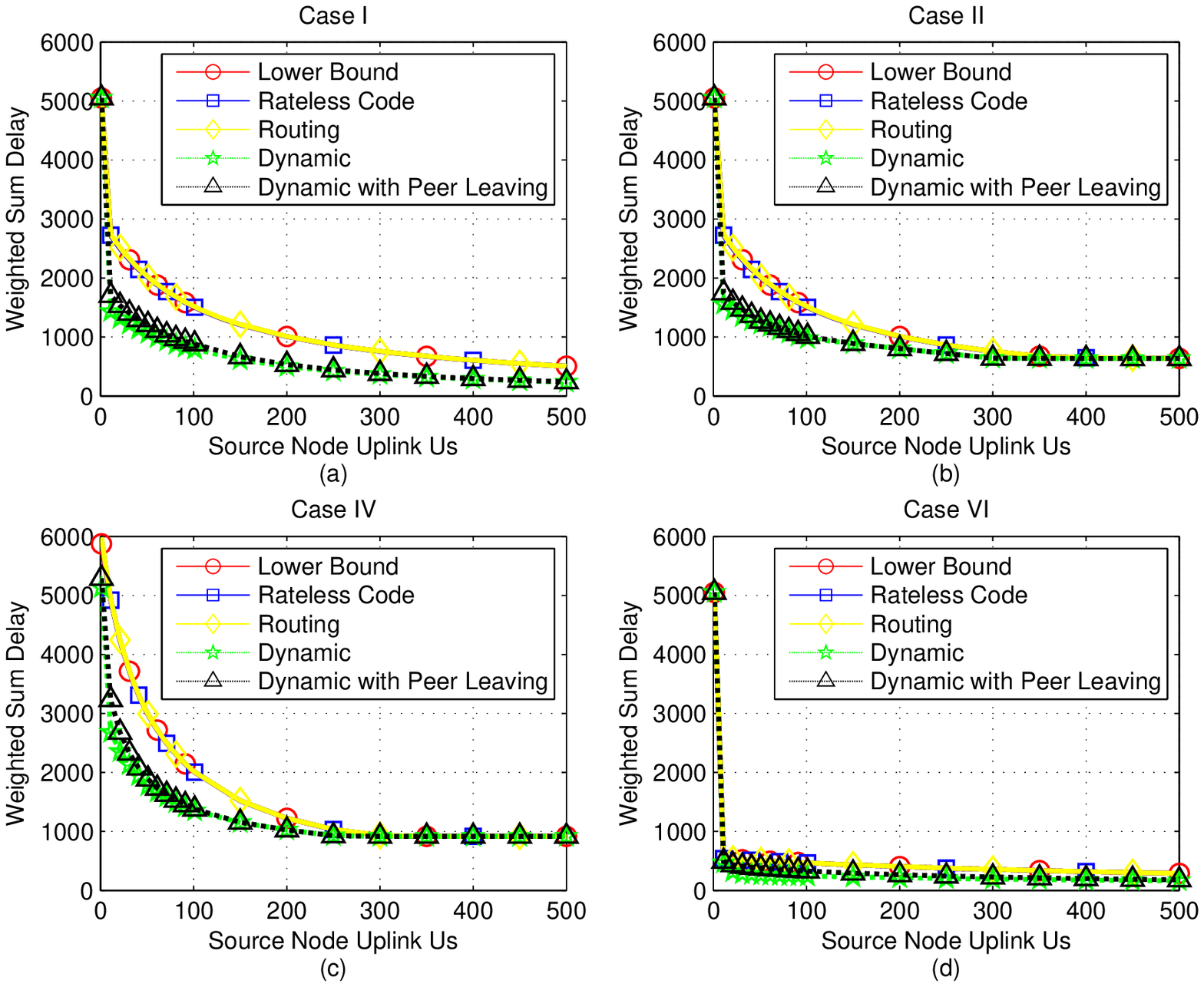}
  \caption{Weighted sum downloading time versus $U_s$ for large P2P networks with $N=100$ peers and weight $W_i = 1 + 99\delta(i>N/2)$.}\label{fig:dynamicsimulation3}
\end{figure}

\begin{figure}
  \centering
  \includegraphics[width=0.7\textwidth]{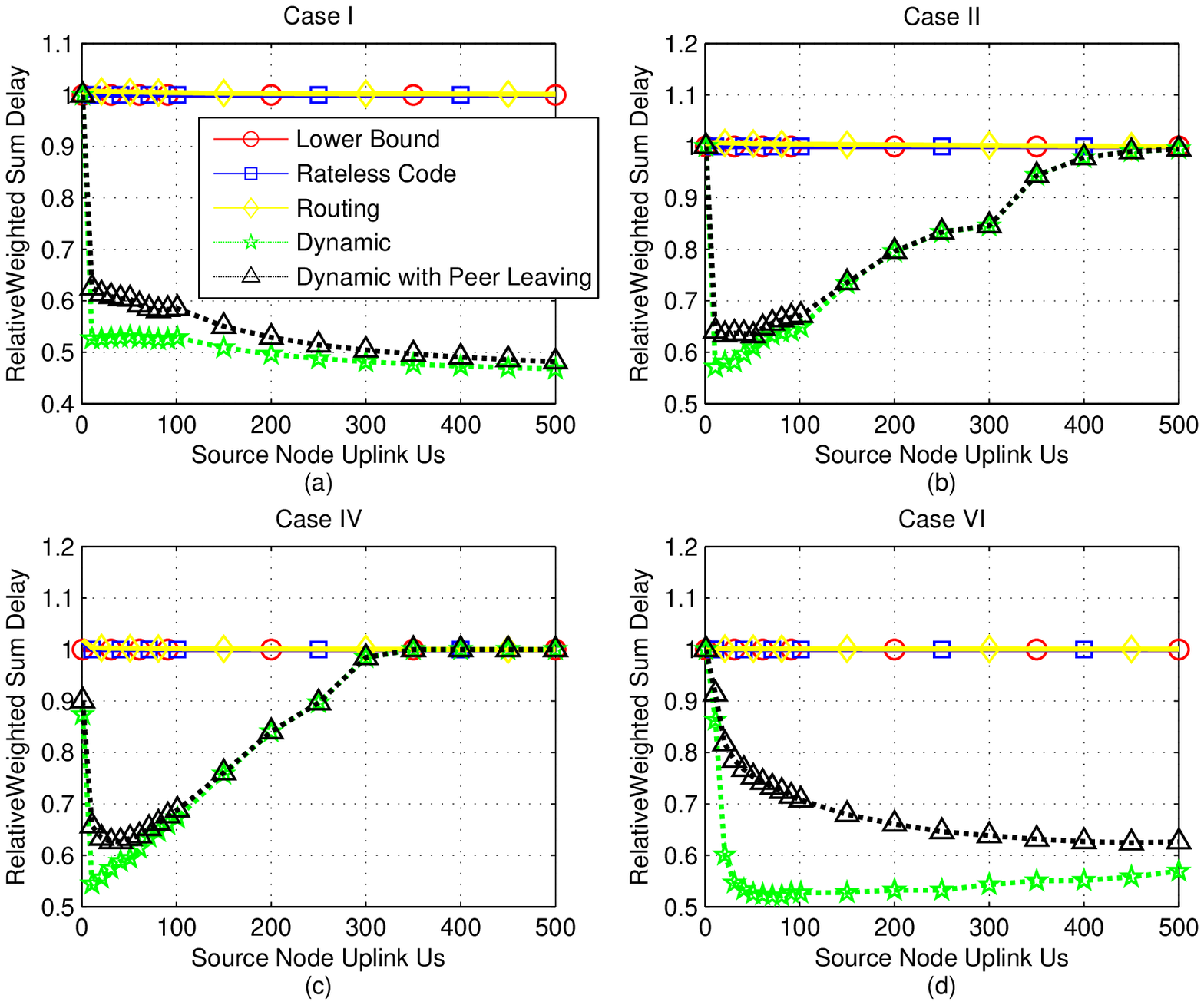}
  \caption{Relative weighted sum downloading time versus $U_s$ for large P2P networks with $N=100$ peers and weight $W_i = 1 + 99\delta(i>N/2)$.}\label{fig:dynamicsimulation13}
\end{figure}

Consider large P2P networks with $N=1000$ peers. The performances of
sum download time versus $U_s$ for the 4 cases are shown in
Fig.~\ref{fig:dynamicsimulation1}.
Fig.~\ref{fig:dynamicsimulation11} shows the relative value of the
sum download time. For Case I, the sum download time of the dynamic
rateless-coding-based scheme is around 0.55 of the minimum sum
download time for the static scenario for a broad range of the
source node uplink $U_s$. The sum download time of the dynamic
rateless-coding-based scheme with peer leaving decreases to 0.70 of
the minimum sum download time for the static scenario as $U_s$
increases to 1000. The WSDT of the dynamic scheme can be as small as
0.57, 0.70, and 0.70 of the minimum WSDT for the static scenario for
Cases II, IV and VI, respectively.

\begin{figure}
  \centering
  \includegraphics[width=0.7\textwidth]{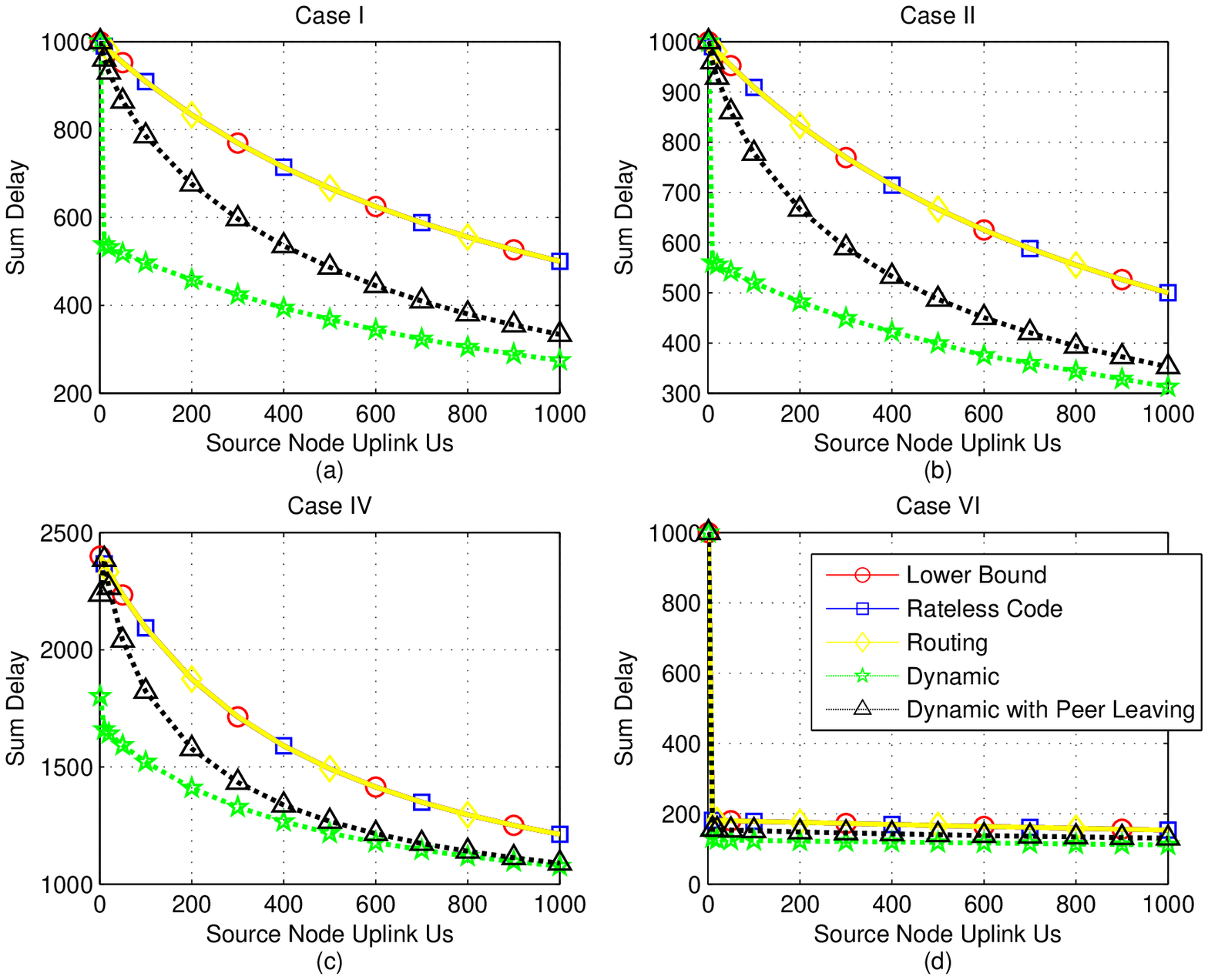}
  \caption{Sum download time versus $U_s$ for large P2P networks with $N=1000$ peers.}\label{fig:dynamicsimulation4}
\end{figure}

\begin{figure}
  \centering
  \includegraphics[width=0.7\textwidth]{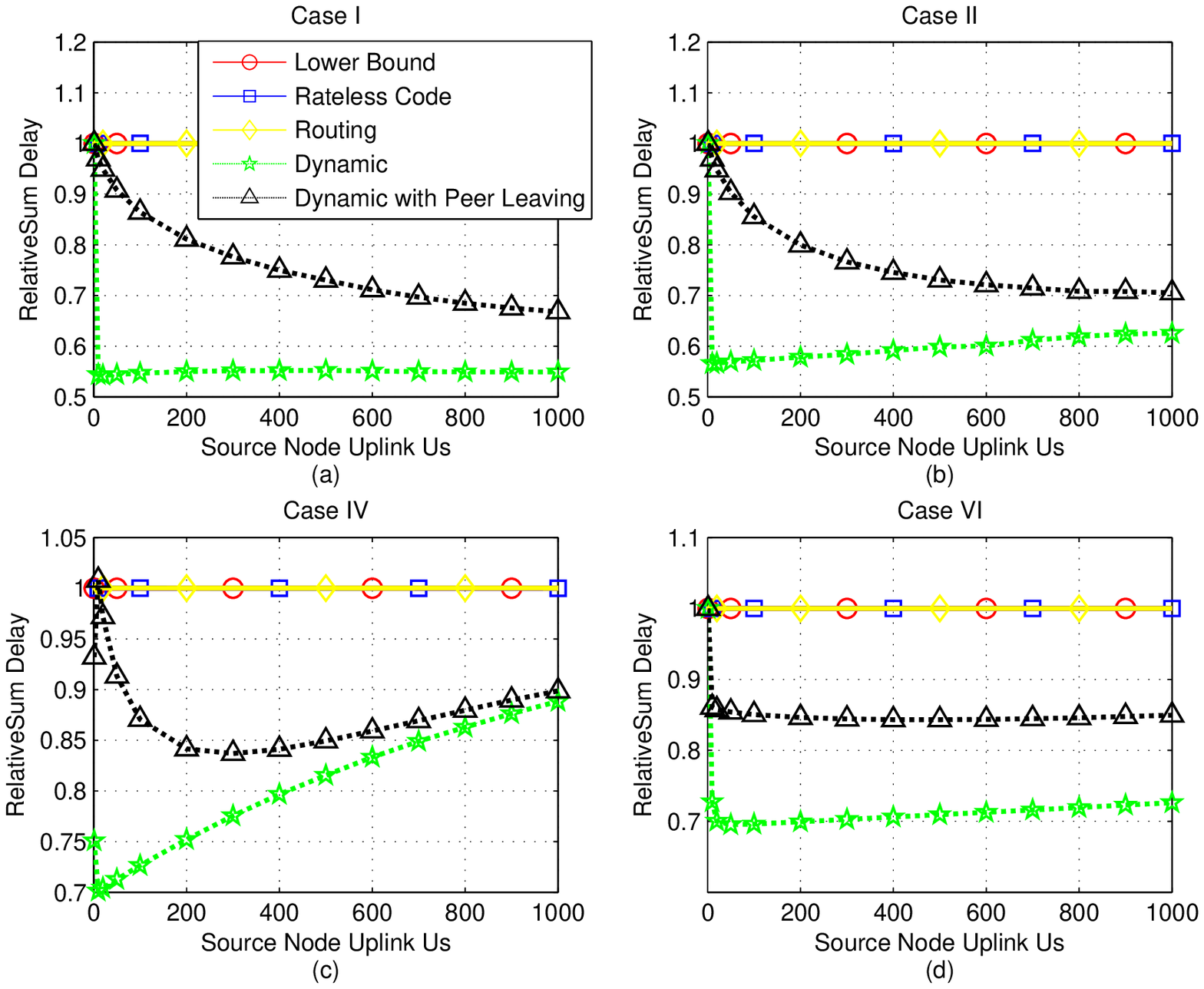}
  \caption{Relative sum download time versus $U_s$ for large P2P networks with $N=1000$ peers.}\label{fig:dynamicsimulation14}
\end{figure}

\section{Conclusion}
\label{sec:conclusions}

This paper considers the problem of transferring a file from one
source node to multiple receivers in a P2P network with both peer
uplink bandwidth constraints and peer downlink bandwidth
constraints.  This paper shows that the static scenario can be
optimized in polynomial time by convex optimization, and the
associated optimal static WSDT can be achieved by linear network
coding. Furthermore, this paper proposes static routing-based  and
rateless-coding-based schemes that closely approach a new lower
bound on performance derived in this paper.

This paper also proposes a dynamic rateless-coding-based scheme,
which provides significantly smaller WSDT than the optimal static
scheme does. A key contribution for the dynamic scenario is a
practical solution to the ordering problem left open by Wu.  Our
solution is to recast this problem as the problem of identifying the
peer weights for each epoch of the ``piecewise static'' rate
allocation.

The deployment of rateless codes simplifies the mechanism of the
file-transfer scenario, enhances the robustness to packet loss in
the network, and increases the performance (without considering
packet overhead).  However, there still exist several issues for
rateless-coding-based scheme such as high source node encoding
complexity, packet overhead, and fast peer selection algorithm for
the dynamic scenario. The results of this paper open interesting
problems in applying rateless codes for P2P applications.

The optimal download time region (set of optimal download times) for
one-to-many file transfer in a P2P network can be characterized by a
system of linear inequalities. Hence, minimizing the WSDT for all
sets of peer weights leads to the download time region. The set of
peer weights can also be assigned according to the applications. For
instances, for a file transfer application with multiple classes of
users, assign a weight to each class of users. For an application
with both receivers and helpers, assign weight zero to helpers and
positive weights to receivers. Hence, the results of this paper in
fact apply directly to one-to-many file transfer applications both
with and without helpers.

\end{spacing}
\end{document}